\documentclass[aps,prl,twocolumn,longbibliography,superscriptaddress]{revtex4-1}
\usepackage{amsmath}
\usepackage{amsthm}
\usepackage{latexsym}
\usepackage{amssymb}
\usepackage{bm}
\usepackage{bbm}
\usepackage{appendix}
\usepackage{graphics,epstopdf}
\usepackage{physics}
\usepackage{color}
\usepackage[colorlinks=true,linkcolor=blue,citecolor=green,plainpages=false,pdfpagelabels]{hyperref}
\usepackage[normalem]{ulem}
\usepackage{newlfont}
\usepackage{amsfonts}
\usepackage{amsthm}
\usepackage{times}

\usepackage{graphicx}
\usepackage{epsfig}
\usepackage{mathrsfs}
\usepackage[thinc]{esdiff}
\usepackage{caption}
\usepackage{subcaption}

\usepackage{times}

\newtheorem{corollary}{Corollary}

\newtheorem{theorem}{Theorem}

\DeclareOldFontCommand{\rm}{\normalfont\rmfamily}{\mathrm}

\usepackage{bigints}

\begin{document}

\title{Exact Quantum Speed Limits}

\author{Arun K. Pati}
\email{akpati@iiit.ac.in}
\affiliation{Centre for Quantum Science and Technology (CQST), International Institute of Information Technology, Hyderabad-500032, India}

\author{Brij Mohan}
\email{brijhcu@gmail.com}
\affiliation{Department of Physical Sciences, Indian Institute of Science Education and Research (IISER), Mohali-140306, India}

\author{Sahil}
\email{sahilmd@imsc.res.in}
\affiliation{Optics and Quantum Information Group, The Institute of Mathematical Sciences, CIT Campus, Taramani, Chennai 600113, India}
\affiliation{Homi Bhabha National Institute, Training School Complex, Anushakti Nagar, Mumbai 400085, India}

\author{Samuel L. Braunstein}
 \email{sam.braunstein@york.ac.uk}
\affiliation{Computer Science, University of York, York YO10 5GH, UK}

\begin{abstract}
The traditional quantum speed limits are not attainable for many physical processes, as they tend to be loose and fail to determine the exact time taken by quantum systems to evolve. To address this, we derive exact quantum speed limits for the unitary dynamics of pure-state quantum system that outperform the existing quantum speed limits. Using these exact quantum speed limits, we can precisely estimate the evolution time for two- and higher-dimensional quantum systems. Additionally, for both finite- and infinite-dimensional quantum systems, we derive an improved Mandelstam-Tamm bound for pure states and show that this bound always saturates for any unitary generated by self-inverse Hamiltonians. Furthermore, we show that our speed limits establish an upper bound on the quantum computational circuit complexity. These results will have a significant impact on our understanding of quantum physics as well as rapidly developing quantum technologies, such as quantum computing, quantum control and quantum thermal machines.

\end{abstract}

\maketitle

{ \it Introduction --} Quantum mechanics imposes fundamental limitations on the rate at which quantum systems evolve over time when subjected to external fields or environments, and these limitations are often referred to as quantum speed limits (QSLs). Quantum speed limits provide a lower bound on the time required to transport an initial quantum state to a final quantum state through a physical process. The first QSL was discovered by Mandelstam and Tamm (MT)~\cite{Mandelstam1945}, which provides an operational interpretation of the widely discussed ``energy-time uncertainty relation" as it sets the intrinsic timescale for a quantum system undergoing unitary evolution~\cite{Aharonov1961}. Several years later, another QSL was discovered by Margolus and Levitin (ML), which is distinct from the MT bound and depends on the mean of the generator of time evolution, rather than its variance~\cite{Margolus1998}. Later on, Levitin and Toffoli demonstrated that the unified bound is achievable under suitable evolution~\cite{Levitin2009}. This unified bound is given by
\begin{equation}
    T \geq T_{\text{QSL}}=\max\{T^{\text{MT}},T^{\text{ML}}\},
\end{equation}
where $T^{\text{MT}} =
\hbar \,\Theta_{\rm 0T}/\Delta H$ 
and $T^{\text{ML}} =
\hbar\, \Theta_{\rm 0T}/\langle H \rangle$ are the MT and ML bounds, respectively. Here, we express these bounds in terms
of the Hilbert space angle $\Theta_{\rm 0T}\equiv\arccos(|\langle\psi_0|\psi_T\rangle|)$ 
between the initial state $|\psi_0\rangle$ and the final state $|\psi_T\rangle$ (instead of twice this angle, e.g., \cite{Anandan1990}), $\Delta H$ is the variance in the Hamiltonian and $\langle H \rangle$ is its expectation value both determined for the initial state of the quantum system. Subsequent developments have also generalized the Mandelstam-Tamm bound to time-dependent Hamiltonians~\cite{Pfeifer1993}.

 Recently it has been shown that the evolution speed of a physical system is constrained not only by the energetic cost but also by the topological structure of the underlying dynamics~\cite{Tan2023}. In Ref.~\cite{Mondal2016}, theoretical proposals were presented to measure quantum speed limits. Interestingly, the MT and ML bounds have recently been measured in matter wave interferometry~\cite{Ness}. QSLs were initially studied for the unitary dynamics of quantum systems described by pure or mixed states~\cite{Mandelstam1945,Margolus1998,Anandan1990,Bhattacharyya1983,Vaidman1992,Braunstein996,Campaioli2018,Campaioli2019,Shao2020,Ness2022,Taddei2013,Braunstein1994,Bagchi2022}. Since then, they have been extensively studied for various types of dynamics, including many-body dynamics~\cite{Bukov2019,Brouzos2015}, non-Hermitian dynamics~\cite{Impens2021,Uzdin2012,Alipour2020,Dimpi2022}, open quantum dynamics~\cite{Lan2022,Nakajima2022,Pires2016,Brody2019,Funo2019,Das2021,Campo2013,Deffner2013}, and arbitrary quantum dynamics~\cite{Dimpi2022}. Previously, QSLs have been studied in the Schr\"odinger picture. However, recently QSLs have been studied in the Heisenberg picture, which establishes an upper bound on the rate at which the expectation values of observables change over time~\cite{B.Mohan2022,Pintos2022,Hamazaki2022}. It has also been found that the celebrated Mandelstam-Tamm bound is a special case of the quantum speed limit of observables~\cite{B.Mohan2022}. 

 The study of QSLs not only enhances our theoretical comprehension of fundamental physics but also carries immense practical significance in the rapidly evolving field of quantum technologies, including but not limited to, quantum computation~\cite{Aifer2022}, quantum metrology~\cite{Zwierz2010,Campbell2018}, quantum optimal control~\cite{Caneva2009}, quantum communication~\cite{Murphy2010} and quantum batteries~\cite{B.Mohan2022,Campaioli2017,Mohan2021,Gyhm2022}. The original QSLs are based on the distinguishability of the initial and final states of a quantum system. It is important to note that QSLs have also been formulated for several other quantities of interest, such as operator flow~\cite{Carabba2022}, generation of quantumness~\cite{Jing2016}, generation of correlations~\cite{Pandey2022,Pandey2023}, and quantum resource production and depletion~\cite{Mohan2022}. QSLs have a significant impact on understanding the limitations of various quantum many-body phenomena, such as information scrambling~\cite{vikram2022}, entanglement growth~\cite{Gong2022,Hamilton2023,Shri2022}, complexity growth~\cite{Niklas2022}, and dynamical quantum phase transition (DQPT)~\cite{Heyl2017,Zhou2021} etc. The concept of speed limits is not limited to quantum physics; it is also widely studied in classical systems~\cite{Okuyama2018,Shanahan2018}, thermodynamics~\cite{Tan20232,Yosh2021,Aghion2023}, and quantum gravity~\cite{Liegener2022,Shabir2023}.

The main goal of studying QSLs is to speed up the quantum process and estimate the minimal time required to implement a given quantum operation. However, a major drawback of the QSLs obtained to date is that they barely saturate, even for unitary dynamics of a closed quantum system described by pure states. Hence, they fail to estimate the exact amount of time required to implement a given quantum operation (or time taken by a given quantum process). The same shortcoming also occurs with both the MT and ML bounds. For example, the famous MT bound only saturates for optimal Hamiltonians that transport the quantum system along a geodesic. Such optimal Hamiltonians have a specific form $H_{\rm opt}=\hbar\omega(|\psi\rangle\langle\psi^{\perp}|+h.c.)$~\cite{Carlini2006,Campaioli2020}. It is worth mentioning that all optimal Hamiltonians are self-inverse, but the converse is not always true. The MT bound can be derived in two different ways: (i) using the Robertson-Heisenberg uncertainty relation~\cite{Mandelstam1945} and (ii) using the geometric approach given by Anandan and Aharonov~\cite{Anandan1990}. Since both approaches use inequalities, the MT bound is also an inequality. Recently, a stronger speed limit~\cite{Dimpi20222} has been derived using the stronger uncertainty relation introduced by Maccone and Pati~\cite{Maccone2014}, which is tighter than the MT bound. However, this bound is also an inequality, and finding where it saturates is difficult.

In this letter, we introduce the notion of exact quantum speed limits (QSLs) that hold for arbitrary unitary processes in two-dimensional as well as higher-dimensional quantum systems described by pure states. We refer to it as the ``exact" quantum speed limit because it precisely estimates the time required to implement any unitary process.
The key to the derivation of these exact QSLs is the exact uncertainty relation proved by Hall~\cite{Hall2001}. The exact quantum speed limit for two-dimensional system depends on the shortest path connecting the initial and final points and on the fluctuation of the non-classical part of the Hamiltonian. For higher-dimensional systems, the exact quantum speed limit depends on the Wootters length (obtained by integrating the distinguishability metric) and on the fluctuation of the non-classical part of the Hamiltonian.  The exact speed limits coincide with the MT bound for optimal Hamiltonians because they have no classical part. Additionally, we have derived an additional QSL in the form of an inequality that applies to any finite- or infinite-dimensional quantum system undergoing unitary dynamics described by pure states. Notably, this bound is tighter than the well-known MT bound and interestingly, for any $d$-dimensional quantum system it saturates for unitaries generated by self-inverse Hamiltonians. Moreover, we find a new way to characterise the optimal Hamiltonians (which saturate the MT bound), namely, that all optimal Hamiltonians have vanishing classical part as defined in Ref.~\cite{Hall2001}; this means they are purely non-classical in nature.

{\it Exact Uncertainty Relation--} 
The uncertainty relations are considered as a cornerstone of quantum mechanics. Most of the uncertainty relations derived to date are in the form of inequalities, even for pure states, except for the uncertainty relation given by Hall~\cite{Hall2001}, which is known as the exact uncertainty relation. The exact uncertainty relation for any two observables $A$ and $B$, when the state of a quantum system is described by a pure state $\rho = \ket{\psi}\bra{\psi}$, is given as follows:
\begin{align}\label{exactUR}
    \delta_B A \,\Delta B^{\text{nc}} = \frac{\hbar}{2}.
\end{align}
Here, 
\begin{align}\label{Sam1}
(\delta_B A)^{-2}=\sum_k \frac{\bra{a_k}\frac{i}{\hbar}[B,\rho]\ket{a_k}^2}{\braket{a_k}{\rho| a_k}},
\end{align}
generalizes the quantum Fisher information for observable $A$ (having eigenbasis \{$\ket{a_k}\}$) with regard to infinitesimal changes in the quantum state $\rho$ generated by the operator $B$ \cite{Hall2001}. This therefore provides a measure of dispersion, $\delta_B A$, for $A$. Further, $\Delta B^{\text{nc}}$ is the uncertainty of the so-called non-classical part of observable $B$ in a state $\rho$, which is defined as $B^{\text{nc}}=B-B^{\text{cl}}$, where the classical part is given by
\begin{align}\label{Sam2}
B^{\text{cl}}=\sum_k\ketbra{a_k}{a_k}\frac{\bra{a_k}\frac{1}{2}\{B,\rho\}\ket{a_k}}{\braket{a_k} {\rho|a_k}},
\end{align}
and $\{,\}$ denotes the anti-commutator. It is important to note that here $B^{\rm cl}$ is not defined in a traditional way (which would simply be the diagonal part of $B$ in some reference basis). The classical part of $B$ commutes with $A$ while its non-classical part does not. In this sense, Eq.~(\ref{exactUR}) is an exact uncertainty relation for pure states. For a comprehensive understanding of this exact uncertainty relation refer to Refs.~\cite{Hall2001,Hall20012,Hall20013}, which provide detailed insights into the mathematical formulation and practical implications of this uncertainty relation.

Now we will discuss the main results of this letter, beginning with the derivation of the exact speed limit for arbitrary unitary dynamics of two-dimensional quantum systems, as well as dynamics of $d$-dimensional quantum systems.
\begin{theorem}\label{Th1}
For a two-dimensional quantum system, the time required to transport a given state $|\psi_0\rangle$ to a target state $|\psi_T\rangle$ via unitary dynamics generated by the Hamiltonian $H_t$ is given by
\begin{equation}\label{exactQSL2d}
    T= \frac{\hbar\, \Theta_{\rm 0T}}{\langle\!\langle \Delta H^{\text{\rm nc}}_t \rangle\!\rangle_{T}}
\end{equation}
whenever the evolution makes $p_t\equiv |\!\braket{\psi_0}{\psi_t}\!|^2$ a monotonic function. Here $\Theta_{\rm 0T}$ is the Hilbert space angle between the initial and final states, $|\psi_0\rangle$ and $|\psi_T\rangle$, respectively, $H^{\text{\rm nc}}_t$ is computed using the basis 
$\{|\psi_0\rangle,|\psi_0^\perp\rangle\}$ and
$\langle\!\langle X_t \rangle\!\rangle_{T}\equiv \frac{1}{T}\!\int_{0}^{T}\! X_t \,dt$  is time average of the quantity $X_t$.

\end{theorem}

\begin{proof}
Let us consider a two-dimensional quantum system which is described by pure states, and whose time evolution is governed by the Hamiltonian $H_t$. Let us take $A=\Pi_{\psi_0}=\ketbra{\psi_0}{\psi_0}$, $B=H_{t}$ and $\rho_{t}=|\psi_{t}\rangle\langle\psi_{t}|$ (where $|\psi_{t}\rangle=U_{t}|\psi_{0}\rangle$) in the exact uncertainty relation Eq.~\eqref{exactUR}. Then, using the definition of $\delta_{H_t}\Pi_{\psi_0}$ defined in Eq.~\eqref{Sam1}, we have
    \begin{align}\label{exactQSL2d1man}
        (\delta_{H_t}\Pi_{\psi_0})^{-1}=-\frac{\frac{d}{dt}|\braket{\psi_0}{\psi_t}|^2}{\sqrt{|\braket{\psi_0}{\psi_t}|^2-|\braket{\psi_0}{\psi_t}|^4}},
    \end{align}
where we take the negative root by assuming $p_t$ is monotonically decreasing.
To arrive at Eq.~\eqref{exactQSL2d1man}, we use the fact that
$|\braket{\psi_0^{\perp}}{\psi_t}|^2=1-|\braket{\psi_0}{\psi_t}|^2=1-p_t$, where \{$\ket{\psi_0},\ket{\psi_0^{\perp}}$\} is the orthonormal basis in two-dimensional Hilbert space. Now the exact uncertainty relation \eqref{exactUR} becomes (for a detailed proof see the Supplementary Material)
\begin{align}\label{exactQSL2d2}
    \Delta H_t^{\text{nc}}
    = -\frac{\hbar}{2}\;\frac{\frac{d}{dt} p_t}{\!\!\sqrt{p_t(1-p_t)}}.
\end{align}
Integrating both sides of this equation then yields
\begin{align}\label{exactQSL2d3}
    T= \frac{\hbar\, \Theta_{\rm 0T}}{\langle\!\langle \Delta H^{\text{\rm nc}}_t \rangle\!\rangle_{T}}.
\end{align}

\end{proof}

The exact speed limit for a two-dimensional quantum system holds for both time-dependent and time-independent Hamiltonians, as long as the survival probability $p_t$ monotonically decreases during evolution. Theorem~\ref{Th1} also holds for $d$-dimensional quantum systems if the time-evolved state belongs to an effective two-dimensional Hilbert space. Since theorem~\ref{Th1} is applicable for a particular type of unitary evolution of $d$-dimensional quantum systems, we will now derive the exact speed limit for arbitrary unitary evolution of $d$-dimensional quantum systems.

\begin{theorem}\label{Th2}
For a d-dimensional quantum system, the time required to transport a given state $|\psi_0\rangle$ to a target state $|\psi_T\rangle$ via unitary dynamics generated by the Hamiltonian $H_t$ is given by
    \begin{align}\label{exactQSLdd}
        T=\frac{\hbar\; l(\phi_t)\vert^T_0}{\langle\!\langle \Delta H^{\text{\rm nc}}_t \rangle\!\rangle_{T}},
    \end{align}
when $H^{\text{\rm nc}}_t$ is computed using the basis 
$\{|a_i\rangle\}$ (where the initial state $|\psi_0\rangle$ is included in this basis) and writing $|\psi_t\rangle=\sum_{i=0}^{d-1} c_i|a_i\rangle$, we may define the length of the path in a real
Hilbert space as $l(\phi_t)\vert^T_0=\bigintsss_{0}^{T}\!\!\!\sqrt{\langle{\dot{\phi_t}}|\dot{\phi_t}\rangle}\,dt$ traced out by the real vector $\ket{\phi_t}=\sum_{i=0}^{d-1}\vert c_i\vert\ket{a_i}$, and $|\dot{\phi_t}\rangle\equiv \frac{d}{dt}\ket{\phi_t}$.
\end{theorem}

\begin{proof}
Let us consider a $d$-dimensional quantum system which is described by pure states, and whose time evolution is governed by the Hamiltonian $H_t$. Using similar reasoning to that used in theorem~\ref{Th1}, we may write
    \begin{align}\label{exQSLdd3}
         (\Delta H_t^{\text{nc}})^2 dt^2 = \frac{\hbar^2}{4}\sum_{i=0}^{d-1}\frac{\braket{a_i}{d\rho_t|a_i}^2}{\braket{a_i}{\rho_t|a_i}}.
    \end{align}
    Now, $\braket{a_i}{d\rho_t|a_i}=d\braket{a_i}{\rho_t|a_i}$ as $\ket{a_i}$ is time independent and hence
    \begin{align}\label{exQSLdd4}
        (\Delta H_t^{\text{nc}})^2 dt^2&=\frac{\hbar^2}{4}\sum_{i=0}^{d-1}\frac{(d|c_i|^2)^2}{|c_i|^2}
        =\hbar^2\sum_{i=0}^{d-1}\left(d|c_i|\right)^2,
    \end{align}
    where $c_i=\braket{a_i}{\psi_t}$. Integrating Eq.~\eqref{exQSLdd4} yields 
    \begin{align}\label{exQSLdd5}
        T=\frac{\hbar}{\langle\!\langle \Delta H^{\text{\rm nc}}_t \rangle\!\rangle_{T}}\bigintss_{0}^{T}\!\!\!\sqrt{\sum_{i=0}^{d-1}(d|c_i|)^2},
    \end{align}
where $\langle\!\langle \Delta H^{\text{\rm nc}}_t \rangle\!\rangle_{T}$ is the time average of $\Delta H^{\text{\rm nc}}_t$.

Defining the vector $\ket{\phi_t}=\sum_{i=0}^{d-1}|c_i|\ket{a_i}$ in a real Hilbert space, its time derivative is given as $|\dot{\phi_t}\rangle=\sum_{i=0}^{d-1}\frac{d\vert c_i\vert }{dt}\ket{a_i}$. Hence, we have $\langle\dot{\phi_t}|\dot{\phi_t}\rangle=\sum_{i=0}^{d-1}\bigl(\frac{d|c_i|}{dt}\bigr)^2$. Now the length of the curve traced out by vector $\ket{\phi_t}$ in real Hilbert space during the evolution is given by
    \begin{align}\label{exQSLdd6}    l(\phi_t)|^{T}_0&=\int_{0}^{T}\!\!\!\sqrt{\langle\dot{\phi_t}|\dot{\phi_t}\rangle}\,dt\nonumber\\
    &=\bigintss_{0}^{T}\!\!\!\sqrt{\sum_{i=0}^{d-1}\,\Bigl(\frac{d|c_i|}{dt}\Bigr)^2}\;dt .
    \end{align}
    Using Eqs.~\eqref{exQSLdd6} and \eqref{exQSLdd5}, we arrive at Eq.~\eqref{exactQSLdd}.
\end{proof}

It is worth mentioning that $l(\phi_t)\vert^T_0$ is exactly equal to the Wootters length $S_{{\text{PD}}}$ on the space of probability distributions (PD). $S_{{\text{PD}}}$ is defined as the integral of the Wootters metric on the space of probability distributions, i.e., $S_{{\text{PD}}}=\int_{\gamma}dS_{{\text{PD}}}$, where $(dS_{\text{PD}})^2 = \frac{1}{4}\sum_{i=0}^{d-1} (dp_i)^2/p_i =\frac{1}{4}F^{\text{cl}}\,dt^2$, $p_i = |c_i|^2$, and $\gamma$ is the curve traced out by a quantum system on the space of probability distributions when observed in some basis and $F^{\text{cl}}$  is just the classical Fisher information with respect to that basis 
$\{\ket{a_i}\}_{i=0}^{d-1}$~\cite{Harry2023,Braunstein1994}. For two-dimensional quantum systems, i.e., $d=2$, $l(\phi_t)\vert^T_0$ reduces to ${\Theta_{\rm 0T}}$ provided the path $\gamma$ does not retrace itself. Therefore theorem~\ref{Th1} may be considered a special case of theorem~\ref{Th2}.

Theorem~\ref{Th2} accurately determines the time required for arbitrary unitary dynamics of a $d$-dimensional system. However, calculating $l(\phi_t)\vert^T_0$ can be challenging compared to ${\Theta_{\rm 0T}}$, which is why we derive another speed limit for both finite- and infinite-dimensional quantum systems. However, it turns out to be an inequality.

\begin{theorem}\label{Th3}
For finite- or infinite-dimensional quantum systems, the time required to transport a given state $|\psi_0\rangle$ to a target state $|\psi_{T}\rangle$ via unitary dynamics generated by the Hamiltonian $H_t$ is lower bounded by the following inequality
\begin{equation}\label{inexactQSL}
   T \geq T_{\text{\rm QSL}}^{\text{\rm IMT}}\equiv \frac{\hbar\, \Theta_{\rm 0T}}{\langle\!\langle \Delta H^{\text{nc}}_t \rangle\!\rangle_{T}}\ge \frac{\hbar\, \Theta_{\rm 0T}}{\langle\!\langle \Delta H_t \rangle\!\rangle_{T}}\equiv T^{\text{\rm MT}}
\end{equation}
where $H_t^{\text{nc}}$ is determined using an orthogonal basis that
includes the initial state and $T^{\text{\rm MT}}$ is the time-dependent generalized form of the Mandelstam-Tamm (MT) bound. (We therefore call the left-hand inequality of Eq.~(\ref{inexactQSL}) the improved MT bound for the quantum speed limit, $T_{\text{\rm QSL}}^{\text{\rm IMT}}$).
\end{theorem}

\begin{proof}
Let us consider finite- or infinite-dimensional quantum systems, which are described using pure states. The time evolution of a given quantum system can be described by the Liouville-von Neumann equation, given as
\begin{equation}
\frac{d\rho_t}{dt}=-\frac{i}{\hbar}[H_t,\rho_t],
\end{equation}
where $H_t$ is the driving Hamiltonian of the quantum system, $\rho_{t} =\ketbra{\psi_t}{\psi_t}$, and $|\psi_{t}\rangle=U_{t}|\psi_{0}\rangle$ is the time-evolved initial state $|\psi_{0}\rangle$.

Taking the expectation of this equation in the initial state then yields
 \begin{align}\label{inexactQSL1}
    \frac{d}{dt}\big|\!\braket{\psi_0}{\psi_t}\!\big|^2= -\frac{i}{\hbar}\Tr(\rho_0[H_t,\rho_t])= \frac{i}{\hbar}\Tr(\rho_t[H_t,\rho_0]),
 \end{align}
As $H_t$ can be expressed as $H_t=H_t^{\text{nc}}+H_t^{\text{cl}}$, and our orthonormal eigenbasis, \{$\ket{\psi_k}\}_{k=0}^{d-1}$, defining these terms involves the initial state, we see that $H^{\text{cl}}_t$ commutes with the initial state, $\rho_0$. Hence, with this basis choice, we have $[H_t,\rho_0]=[H_t^{\text{nc}},\rho_0]$. Next, using Eq.~\eqref{inexactQSL1} and the Robertson-Heisenberg uncertainty relation \cite{Robertson1929}, $\frac{1}{2}|\langle[A,B]\rangle|\le \Delta A\Delta B$, we obtain 
\begin{align}\label{inexactQSL2}
    \Big|\frac{d}{dt}|\braket{\psi_0}{\psi_t}|^2\Big|= \big|\frac{i}{\hbar}{\Tr}(\rho_t[H^{\text{nc}}_t,\rho_0])\big| \leq \frac{2}{\hbar}\,\Delta \rho_0 \,\Delta H_t^{\text{nc}}.
 \end{align}
where $(\Delta \rho_0)^2 \equiv {\Tr}(\rho_t\, \rho_0^2)-[{\Tr}(\rho_t \,\rho_0)]^2=p_t-p_t^2$ and $p_t={\Tr}(\rho_t\, \rho_0)=|\!\braket{\psi_0}{\psi_t}\!|^2$. Consequently, Eq.~\eqref{inexactQSL2} becomes
\begin{align}\label{inexactQSL3}
         \Delta H_t^{\text{nc}}dt\geq \frac{\hbar}{2}\,\frac{|dp_t|}{\!\!\sqrt{p_t-p_t^2}}.
\end{align}
Integrating Eq.~\eqref{inexactQSL3} and rearranging it now yields
     \begin{align}   
        T &\geq \frac{\hbar}{2\langle\!\langle \Delta H^{\text{\rm nc}}_t \rangle\!\rangle_{T}}\int_{0}^{T}\!\!\frac{|dp_t|}{\!\!\sqrt{p_t-p_t^2}}\nonumber\\
        &\geq \frac{\hbar}{2\langle\!\langle \Delta H^{\text{\rm nc}}_t \rangle\!\rangle_{T}}\left|\int_{0}^{T}\!\!\frac{dp_t}{\!\!\sqrt{p_t-p_t^2}}\right|\nonumber\\
        &=\frac{\hbar\, \Theta_{\rm 0T}}{\langle\!\langle \Delta H^{\text{\rm nc}}_t \rangle\!\rangle_{T}}\equiv T_{\text{QSL}}^{\text{IMT}},\nonumber
     \end{align}
where we have used the triangle inequality for the integral to obtain the last inequality.

To finish off the proof, we need only show that $T_{\text{QSL}}^{\text{IMT}}\geq T^{\text{MT}}$, with the time-dependent version of $T^{\text{MT}}$ given in the statement of the theorem. Recall that
\begin{align}
(\Delta H_t)^2=(\Delta H^{\text{\rm nc}}_t)^2
+(\Delta H^{\text{\rm cl}}_t)^2\ge (\Delta H^{\text{\rm nc}}_t)^2,
\end{align}
hence $\langle\!\langle \Delta H^{\text{\rm nc}}_t \rangle\!\rangle_{T} \leq \langle\!\langle \Delta H_t \rangle\!\rangle_{T}$ and the result follows.
\end{proof}

This theorem yields an improvement over the usual MT bound by removing the unnecessary `classical' contribution to the uncertainty in the energy. It can be shown that Hamiltonians, $H_{\text{opt}}$, that saturate the MT bound will also saturate the improved MT bound. In this case, we have $T_{\text{QSL}}^{\text{IMT}}= T^{\text{MT}}$ because such Hamiltonians have a vanishing classical part (i.e., $H_{\text{opt}}^{\text{cl}}=0$). Further, we may identify the quantity $\Delta H_t^{\text{nc}}/\hbar$ as the actual speed of evolution of a pure quantum system in the Hilbert space, improving on Ref.~\cite{Anandan1990}'s claim that the speed of evolution is $\Delta H_t/\hbar$ (with distances, in each case, measured via the Fubini-Study metric~\cite{Anandan1990}).

We note that while all optimal Hamiltonians for the standard MT bound, $H_{\text{opt}}$, are self-inverse, the converse is not true. In general, $\Delta H_t^{nc}$ turns out to be time-dependent even for time-independent Hamiltonians. However, $\Delta H_t^{nc}$ is time-independent when the Hamiltonian is optimal. We now show that all self-inverse Hamiltonians also saturate the improved MT bound.

\begin{corollary}\label{cor}
The bound derived in Theorem~\ref{Th3} saturates for the  unitary dynamics generated by time-independent self-inverse Hamiltonians, i.e., for which $H^2=I$ (where we set $\hbar=1$ and make the Hamiltonian dimensionless).
\end{corollary}
See the Supplementary Material for the detailed proof of the above corollary.
While proving this corollary we made no assumptions about the dimension of the system or the initial state. Therefore, it equally holds for $d$-dimensional systems. It is important to mention the significance of self-inverse Hamiltonians because they are highly useful in both single- and multi-qubit systems due to their ability to describe a wide range of physical interactions between qubits.  Additionally, self-inverse Hamiltonians can be used to implement a variety of single-qubit and two-qubit quantum gates, including rotations, phase shifts, and the CNOT gate. Therefore the bound~\eqref{inexactQSL} is highly useful for quantum computing. Moreover, the bound~\eqref{inexactQSL} can be saturated for a more general class of Hamiltonians, but finding such Hamiltonians is an extremely difficult task.\par

It is worth mentioning that the unitary dynamics for which the improved MT bound saturates, Eq.~\eqref{inexactQSL}, may not be the fastest dynamics. To identify the fastest unitary dynamics, we need to estimate the evolution time and compare it with the time taken by other Hamiltonians to reach the target state from a given initial state. In short, we need to minimize the exact evolution time determined by exact QSLs over all possible Hamiltonians that connect these states, i.e., $T_{\text{opt}}=\min_{H}T$.
Estimating the speed limits presented in this letter requires information about the time-evolved state and generator of the dynamics (i.e., the driving Hamiltonian), which is slightly more involved than calculating the MT bound or ML bound for a time-independent Hamiltonian, as the latter only requires knowledge of the generator of the dynamics. However, for time-dependent Hamiltonians, the MT bound requires knowledge of the time-evolved state which would still not allow us to improve the MT bound itself. Also, all other QSLs for any dynamics introduced to date require the same amount of information as our exact QSLs, but they fail to accurately estimate the exact evolution time. In this sense, our exact quantum speed limits outperform previously obtained results.

{\it  Relation between Exact Speed limit, Parameter estimation and Circuit Complexity:}
If $\theta$ is an unknown parameter encoded by the unitary operator $U_{\theta}=e^{-\frac{iH\theta}{\hbar}}$ in the quantum state, then, without loss of generality, the bound obtained in theorem 3 can be generalized as follows:
\begin{equation}
   \theta \geq  \frac{\hbar\, \Theta_{\rm 0 \rm \theta}}{\langle\!\langle \Delta H^{\text{nc}} \rangle\!\rangle_{\rm \theta}}\ge
   \frac{\hbar\, \Theta_{\rm 0 \rm \theta}}{\langle\!\langle \Delta H \rangle\!\rangle_{\rm \theta}} =
   \frac{\hbar\, \Theta_{\rm 0 \rm \theta}}{\Delta H}\;,
   \label{unknownP}
\end{equation}
where the final step relies on the fact that
$[H,U_\theta]=0$.

In quantum computation, one of the challenges is to find an efficient circuit for implementing a unitary operation $U_{\theta}=e^{-\frac{iH\theta}{\hbar}}$ (where $\theta$ is a known parameter), which can be used to solve computational problems like Grover's search algorithm or Shor's factoring algorithm. The minimum number of gates needed to implement a specific unitary transformation, taking the system from a reference state $|\psi_{0}\rangle$ to a target state $|\psi_{\rm \theta}\rangle$, is known as its circuit complexity~\cite{Jeff2017,Chapman2018,Mic2006}. 
 This is now widely used in many other areas including high energy physics~\cite{Chapman2018,Caputa2019,Adhikari2021} and condensed matter physics~\cite{Liu2020}. It has been found that the geodesic distance $d(I,U_{\theta})$ (the shortest distance between the Identity operator  and the Unitary operator of interest) serves as a good measure of circuit complexity. Therefore, the circuit complexity for a given reference state may be defined as~\cite{Adhikari2021,Chapman2018,Brown2019,Poggi2019}:
\begin{equation}
\mathcal{C} \equiv d(I,U_{\theta}) = \arccos{|\langle\psi|U_{\theta}|\psi\rangle|}
=\Theta_{0\theta}\;.
\end{equation}
Interestingly, the circuit complexity reduces to the Hilbert space angle. Thus, the following relation can be derived using the above bound of Eq.~(\ref{unknownP}):
\begin{equation}
    \cal{C} \leq \frac{\theta\,{\langle\!\langle \Delta H^{\text{\rm nc}}\rangle\!\rangle_{\theta}} }{\hbar}   \leq  \frac{\theta\, \Delta H}{\hbar} .
\end{equation}
The above inequality provides an upper bound on computational complexity. For an optimal Hamiltonian, $H_{\rm opt}$, this bound saturates, implying that the circuit complexity may be precisely estimated as $\mathcal{C} = \omega\,\theta$, where $\omega$ is the coupling strength of $H_{\rm opt}$.

{\it Conclusions--} In this letter, we have derived the exact speed limit for two-dimensional closed quantum systems as well as higher-dimensional quantum systems described by pure states. The exact quantum speed limits are the consequence of a novel exact uncertainty relation for pure states introduced by Hall~\cite{Hall2001,Hall20012,Hall20013}. Additionally, we have derived a quantum speed limit for both finite- and infinite-dimensional quantum systems and show that if the generator of dynamics is self-inverse the bound saturates. In support of our findings, we also provided an example of a qubit system (see Supplementary Material). Moreover, we have identified the key reason why the MT bound does not saturate for most unitary processes: the generators (the driving Hamiltonians) contain some classical contribution. To obtain exact QSLs for most unitary processes, it is necessary to remove this `classicality' from the generators of the dynamics because it commutes with the state and does not contribute to its evolution. The QSLs derived here can be seen as an improved version of the MT bound because they depend on the variance of the non-classical part of the generator rather than the variance of the generator itself. The exact QSLs are solely dependent on the variance of the non-classical part of generator since it is the only component responsible for the system's evolution. The non-classical part of the generator strongly depends on the initial and time-evolved state of the given quantum system. Our exact quantum speed limits can accurately estimate the time required to perform any given task using quantum devices. Moreover, the exact speed limits can be thought of as an exact time-energy uncertainty relation.

Since the exact QSLs have been derived using fundamentally well-established concepts (non-commutativity of quantum mechanical observables and uncertainty relations), therefore, we believe that they can be easily verified in future experiments. Our results will significantly enhance our understanding of the time-energy uncertainty relation in quantum physics. Moreover, we believe that these results will have numerous applications in the rapidly developing area of quantum technologies including quantum computation, quantum control, quantum communication, quantum thermal machines and quantum energy storage devices. The ideas presented in this letter open the door to finding exact speed limits for mixed states as well as for open dynamical systems. Moreover, this approach can be extended beyond the standard speed limits because it can be helpful to derive the exact or tight speed limits for other quantities such as observables, operator flow, entanglement growth, information scrambling, and the growth of operator complexity.

\vskip 0.1in
\begin{acknowledgements}
B. Mohan and Sahil express their gratitude to Samyadeb Bhattacharya for facilitating their visit to CQST, IIIT Hyderabad. B. Mohan and Sahil acknowledge Siddhartha Das and Uttam Singh for the fruitful discussion. A.K. Pati acknowledges the support of the J. C. Bose Fellowship from the Department of Science and Technology (DST) India under Grant No.~JCB/2018/000038 for the period 2019-2024.
\end{acknowledgements}

\bibliography{main}

\begin{thebibliography}{80}%
\makeatletter
\providecommand \@ifxundefined [1]{%
 \@ifx{#1\undefined}
}%
\providecommand \@ifnum [1]{%
 \ifnum #1\expandafter \@firstoftwo
 \else \expandafter \@secondoftwo
 \fi
}%
\providecommand \@ifx [1]{%
 \ifx #1\expandafter \@firstoftwo
 \else \expandafter \@secondoftwo
 \fi
}%
\providecommand \natexlab [1]{#1}%
\providecommand \enquote  [1]{``#1''}%
\providecommand \bibnamefont  [1]{#1}%
\providecommand \bibfnamefont [1]{#1}%
\providecommand \citenamefont [1]{#1}%
\providecommand \href@noop [0]{\@secondoftwo}%
\providecommand \href [0]{\begingroup \@sanitize@url \@href}%
\providecommand \@href[1]{\@@startlink{#1}\@@href}%
\providecommand \@@href[1]{\endgroup#1\@@endlink}%
\providecommand \@sanitize@url [0]{\catcode `\\12\catcode `\$12\catcode
  `\&12\catcode `\#12\catcode `\^12\catcode `\_12\catcode `\%12\relax}%
\providecommand \@@startlink[1]{}%
\providecommand \@@endlink[0]{}%
\providecommand \url  [0]{\begingroup\@sanitize@url \@url }%
\providecommand \@url [1]{\endgroup\@href {#1}{\urlprefix }}%
\providecommand \urlprefix  [0]{URL }%
\providecommand \Eprint [0]{\href }%
\providecommand \doibase [0]{http://dx.doi.org/}%
\providecommand \selectlanguage [0]{\@gobble}%
\providecommand \bibinfo  [0]{\@secondoftwo}%
\providecommand \bibfield  [0]{\@secondoftwo}%
\providecommand \translation [1]{[#1]}%
\providecommand \BibitemOpen [0]{}%
\providecommand \bibitemStop [0]{}%
\providecommand \bibitemNoStop [0]{.\EOS\space}%
\providecommand \EOS [0]{\spacefactor3000\relax}%
\providecommand \BibitemShut  [1]{\csname bibitem#1\endcsname}%
\let\auto@bib@innerbib\@empty
\bibitem [{\citenamefont {Mandelstam}\ and\ \citenamefont
  {Tamm}(1945)}]{Mandelstam1945}%
  \BibitemOpen
  \bibfield  {author} {\bibinfo {author} {\bibfnamefont {Leonid}\ \bibnamefont
  {Mandelstam}}\ and\ \bibinfo {author} {\bibfnamefont {IG}~\bibnamefont
  {Tamm}},\ }\bibfield  {title} {\enquote {\bibinfo {title} {The uncertainty
  relation between energy and time in non-relativistic quantum mechanics},}\
  }\href {https://doi.org/10.1007/978-3-642-74626-0_8} {\bibfield  {journal}
  {\bibinfo  {journal} {J. Phys. (USSR)}\ }\textbf {\bibinfo {volume} {9}},\
  \bibinfo {pages} {249} (\bibinfo {year} {1945})}\BibitemShut {NoStop}%
\bibitem [{\citenamefont {Aharonov}\ and\ \citenamefont
  {Bohm}(1961)}]{Aharonov1961}%
  \BibitemOpen
  \bibfield  {author} {\bibinfo {author} {\bibfnamefont {Y.}~\bibnamefont
  {Aharonov}}\ and\ \bibinfo {author} {\bibfnamefont {D.}~\bibnamefont
  {Bohm}},\ }\bibfield  {title} {\enquote {\bibinfo {title} {Time in the
  quantum theory and the uncertainty relation for time and energy},}\ }\href
  {\doibase 10.1103/PhysRev.122.1649} {\bibfield  {journal} {\bibinfo
  {journal} {Physical Review}\ }\textbf {\bibinfo {volume} {122}},\ \bibinfo
  {pages} {1649--1658} (\bibinfo {year} {1961})}\BibitemShut {NoStop}%
\bibitem [{\citenamefont {Margolus}\ and\ \citenamefont
  {Levitin}(1998)}]{Margolus1998}%
  \BibitemOpen
  \bibfield  {author} {\bibinfo {author} {\bibfnamefont {Norman}\ \bibnamefont
  {Margolus}}\ and\ \bibinfo {author} {\bibfnamefont {Lev~B.}\ \bibnamefont
  {Levitin}},\ }\bibfield  {title} {\enquote {\bibinfo {title} {The maximum
  speed of dynamical evolution},}\ }\href {\doibase
  https://doi.org/10.1016/S0167-2789(98)00054-2} {\bibfield  {journal}
  {\bibinfo  {journal} {Physica D: Nonlinear Phenomena}\ }\textbf {\bibinfo
  {volume} {120}},\ \bibinfo {pages} {188--195} (\bibinfo {year}
  {1998})}\BibitemShut {NoStop}%
\bibitem [{\citenamefont {Levitin}\ and\ \citenamefont
  {Toffoli}(2009)}]{Levitin2009}%
  \BibitemOpen
  \bibfield  {author} {\bibinfo {author} {\bibfnamefont {Lev~B.}\ \bibnamefont
  {Levitin}}\ and\ \bibinfo {author} {\bibfnamefont {Tommaso}\ \bibnamefont
  {Toffoli}},\ }\bibfield  {title} {\enquote {\bibinfo {title} {Fundamental
  {L}imit on the {R}ate of {Q}uantum {D}ynamics: The {U}nified {B}ound {I}s
  {T}ight},}\ }\href {\doibase 10.1103/PhysRevLett.103.160502} {\bibfield
  {journal} {\bibinfo  {journal} {Physical Review Letters}\ }\textbf {\bibinfo
  {volume} {103}},\ \bibinfo {pages} {160502} (\bibinfo {year}
  {2009})}\BibitemShut {NoStop}%
\bibitem [{\citenamefont {Anandan}\ and\ \citenamefont
  {Aharonov}(1990)}]{Anandan1990}%
  \BibitemOpen
  \bibfield  {author} {\bibinfo {author} {\bibfnamefont {J.}~\bibnamefont
  {Anandan}}\ and\ \bibinfo {author} {\bibfnamefont {Y.}~\bibnamefont
  {Aharonov}},\ }\bibfield  {title} {\enquote {\bibinfo {title} {Geometry of
  quantum evolution},}\ }\href {\doibase 10.1103/PhysRevLett.65.1697}
  {\bibfield  {journal} {\bibinfo  {journal} {Physical Review Letters}\
  }\textbf {\bibinfo {volume} {65}},\ \bibinfo {pages} {1697--1700} (\bibinfo
  {year} {1990})}\BibitemShut {NoStop}%
\bibitem [{\citenamefont {Pfeifer}(1993)}]{Pfeifer1993}%
  \BibitemOpen
  \bibfield  {author} {\bibinfo {author} {\bibfnamefont {Peter}\ \bibnamefont
  {Pfeifer}},\ }\bibfield  {title} {\enquote {\bibinfo {title} {How fast can a
  quantum state change with time?}}\ }\href {\doibase
  10.1103/PhysRevLett.70.3365} {\bibfield  {journal} {\bibinfo  {journal}
  {Physical Review Letters}\ }\textbf {\bibinfo {volume} {70}},\ \bibinfo
  {pages} {3365--3368} (\bibinfo {year} {1993})}\BibitemShut {NoStop}%
\bibitem [{\citenamefont {Van~Vu}\ and\ \citenamefont
  {Saito}(2023{\natexlab{a}})}]{Tan2023}%
  \BibitemOpen
  \bibfield  {author} {\bibinfo {author} {\bibfnamefont {Tan}\ \bibnamefont
  {Van~Vu}}\ and\ \bibinfo {author} {\bibfnamefont {Keiji}\ \bibnamefont
  {Saito}},\ }\bibfield  {title} {\enquote {\bibinfo {title} {Topological speed
  limit},}\ }\href {\doibase 10.1103/PhysRevLett.130.010402} {\bibfield
  {journal} {\bibinfo  {journal} {Physical Review Letters}\ }\textbf {\bibinfo
  {volume} {130}},\ \bibinfo {pages} {010402} (\bibinfo {year}
  {2023}{\natexlab{a}})}\BibitemShut {NoStop}%
\bibitem [{\citenamefont {Mondal}\ and\ \citenamefont
  {Pati}(2016)}]{Mondal2016}%
  \BibitemOpen
  \bibfield  {author} {\bibinfo {author} {\bibfnamefont {Debasis}\ \bibnamefont
  {Mondal}}\ and\ \bibinfo {author} {\bibfnamefont {Arun~Kumar}\ \bibnamefont
  {Pati}},\ }\bibfield  {title} {\enquote {\bibinfo {title} {Quantum speed
  limit for mixed states using an experimentally realizable metric},}\ }\href
  {\doibase https://doi.org/10.1016/j.physleta.2016.02.018} {\bibfield
  {journal} {\bibinfo  {journal} {Physics Letters A}\ }\textbf {\bibinfo
  {volume} {380}},\ \bibinfo {pages} {1395--1400} (\bibinfo {year}
  {2016})}\BibitemShut {NoStop}%
\bibitem [{\citenamefont {Ness}\ \emph {et~al.}(2021)\citenamefont {Ness},
  \citenamefont {Lam}, \citenamefont {Alt}, \citenamefont {Meschede},
  \citenamefont {Sagi},\ and\ \citenamefont {Alberti}}]{Ness}%
  \BibitemOpen
  \bibfield  {author} {\bibinfo {author} {\bibfnamefont {Gal}\ \bibnamefont
  {Ness}}, \bibinfo {author} {\bibfnamefont {Manolo~R.}\ \bibnamefont {Lam}},
  \bibinfo {author} {\bibfnamefont {Wolfgang}\ \bibnamefont {Alt}}, \bibinfo
  {author} {\bibfnamefont {Dieter}\ \bibnamefont {Meschede}}, \bibinfo {author}
  {\bibfnamefont {Yoav}\ \bibnamefont {Sagi}}, \ and\ \bibinfo {author}
  {\bibfnamefont {Andrea}\ \bibnamefont {Alberti}},\ }\bibfield  {title}
  {\enquote {\bibinfo {title} {Observing crossover between quantum speed
  limits},}\ }\href {\doibase 10.1126/sciadv.abj9119} {\bibfield  {journal}
  {\bibinfo  {journal} {Science Advances}\ }\textbf {\bibinfo {volume} {7}},\
  \bibinfo {pages} {eabj9119} (\bibinfo {year} {2021})}\BibitemShut {NoStop}%
\bibitem [{\citenamefont {Bhattacharyya}(1983)}]{Bhattacharyya1983}%
  \BibitemOpen
  \bibfield  {author} {\bibinfo {author} {\bibfnamefont {K}~\bibnamefont
  {Bhattacharyya}},\ }\bibfield  {title} {\enquote {\bibinfo {title} {Quantum
  decay and the {M}andelstam-{T}amm-energy inequality},}\ }\href {\doibase
  10.1088/0305-4470/16/13/021} {\bibfield  {journal} {\bibinfo  {journal}
  {Journal of Physics A: Mathematical and General}\ }\textbf {\bibinfo {volume}
  {16}},\ \bibinfo {pages} {2993} (\bibinfo {year} {1983})}\BibitemShut
  {NoStop}%
\bibitem [{\citenamefont {Vaidman}(1992)}]{Vaidman1992}%
  \BibitemOpen
  \bibfield  {author} {\bibinfo {author} {\bibfnamefont {Lev}\ \bibnamefont
  {Vaidman}},\ }\bibfield  {title} {\enquote {\bibinfo {title} {Minimum time
  for the evolution to an orthogonal quantum state},}\ }\href
  {https://doi.org/10.1119/1.16940} {\bibfield  {journal} {\bibinfo  {journal}
  {American journal of physics}\ }\textbf {\bibinfo {volume} {60}},\ \bibinfo
  {pages} {182--183} (\bibinfo {year} {1992})}\BibitemShut {NoStop}%
\bibitem [{\citenamefont {Braunstein}\ \emph {et~al.}(1996)\citenamefont
  {Braunstein}, \citenamefont {Caves},\ and\ \citenamefont
  {Milburn}}]{Braunstein996}%
  \BibitemOpen
  \bibfield  {author} {\bibinfo {author} {\bibfnamefont {Samuel~L.}\
  \bibnamefont {Braunstein}}, \bibinfo {author} {\bibfnamefont {Carlton~M.}\
  \bibnamefont {Caves}}, \ and\ \bibinfo {author} {\bibfnamefont {G.J.}\
  \bibnamefont {Milburn}},\ }\bibfield  {title} {\enquote {\bibinfo {title}
  {Generalized uncertainty relations: Theory, examples, and lorentz
  invariance},}\ }\href {\doibase https://doi.org/10.1006/aphy.1996.0040}
  {\bibfield  {journal} {\bibinfo  {journal} {Annals of Physics}\ }\textbf
  {\bibinfo {volume} {247}},\ \bibinfo {pages} {135--173} (\bibinfo {year}
  {1996})}\BibitemShut {NoStop}%
\bibitem [{\citenamefont {Campaioli}\ \emph {et~al.}(2018)\citenamefont
  {Campaioli}, \citenamefont {Pollock}, \citenamefont {Binder},\ and\
  \citenamefont {Modi}}]{Campaioli2018}%
  \BibitemOpen
  \bibfield  {author} {\bibinfo {author} {\bibfnamefont {Francesco}\
  \bibnamefont {Campaioli}}, \bibinfo {author} {\bibfnamefont {Felix~A.}\
  \bibnamefont {Pollock}}, \bibinfo {author} {\bibfnamefont {Felix~C.}\
  \bibnamefont {Binder}}, \ and\ \bibinfo {author} {\bibfnamefont {Kavan}\
  \bibnamefont {Modi}},\ }\bibfield  {title} {\enquote {\bibinfo {title}
  {Tightening quantum speed limits for almost all states},}\ }\href {\doibase
  10.1103/PhysRevLett.120.060409} {\bibfield  {journal} {\bibinfo  {journal}
  {Physical Review Letters}\ }\textbf {\bibinfo {volume} {120}},\ \bibinfo
  {pages} {060409} (\bibinfo {year} {2018})}\BibitemShut {NoStop}%
\bibitem [{\citenamefont {Campaioli}\ \emph {et~al.}(2019)\citenamefont
  {Campaioli}, \citenamefont {Pollock},\ and\ \citenamefont
  {Modi}}]{Campaioli2019}%
  \BibitemOpen
  \bibfield  {author} {\bibinfo {author} {\bibfnamefont {Francesco}\
  \bibnamefont {Campaioli}}, \bibinfo {author} {\bibfnamefont {Felix~A.}\
  \bibnamefont {Pollock}}, \ and\ \bibinfo {author} {\bibfnamefont {Kavan}\
  \bibnamefont {Modi}},\ }\bibfield  {title} {\enquote {\bibinfo {title}
  {Tight, robust, and feasible quantum speed limits for open dynamics},}\
  }\href {\doibase 10.22331/q-2019-08-05-168} {\bibfield  {journal} {\bibinfo
  {journal} {{Quantum}}\ }\textbf {\bibinfo {volume} {3}},\ \bibinfo {pages}
  {168} (\bibinfo {year} {2019})}\BibitemShut {NoStop}%
\bibitem [{\citenamefont {Shao}\ \emph {et~al.}(2020)\citenamefont {Shao},
  \citenamefont {Liu}, \citenamefont {Zhang}, \citenamefont {Yuan},\ and\
  \citenamefont {Liu}}]{Shao2020}%
  \BibitemOpen
  \bibfield  {author} {\bibinfo {author} {\bibfnamefont {Yanyan}\ \bibnamefont
  {Shao}}, \bibinfo {author} {\bibfnamefont {Bo}~\bibnamefont {Liu}}, \bibinfo
  {author} {\bibfnamefont {Mao}\ \bibnamefont {Zhang}}, \bibinfo {author}
  {\bibfnamefont {Haidong}\ \bibnamefont {Yuan}}, \ and\ \bibinfo {author}
  {\bibfnamefont {Jing}\ \bibnamefont {Liu}},\ }\bibfield  {title} {\enquote
  {\bibinfo {title} {Operational definition of a quantum speed limit},}\ }\href
  {\doibase 10.1103/PhysRevResearch.2.023299} {\bibfield  {journal} {\bibinfo
  {journal} {Physical Review Research}\ }\textbf {\bibinfo {volume} {2}},\
  \bibinfo {pages} {023299} (\bibinfo {year} {2020})}\BibitemShut {NoStop}%
\bibitem [{\citenamefont {Ness}\ \emph {et~al.}(2022)\citenamefont {Ness},
  \citenamefont {Alberti},\ and\ \citenamefont {Sagi}}]{Ness2022}%
  \BibitemOpen
  \bibfield  {author} {\bibinfo {author} {\bibfnamefont {Gal}\ \bibnamefont
  {Ness}}, \bibinfo {author} {\bibfnamefont {Andrea}\ \bibnamefont {Alberti}},
  \ and\ \bibinfo {author} {\bibfnamefont {Yoav}\ \bibnamefont {Sagi}},\
  }\bibfield  {title} {\enquote {\bibinfo {title} {Quantum speed limit for
  states with a bounded energy spectrum},}\ }\href {\doibase
  10.1103/PhysRevLett.129.140403} {\bibfield  {journal} {\bibinfo  {journal}
  {Physical Review Letters}\ }\textbf {\bibinfo {volume} {129}},\ \bibinfo
  {pages} {140403} (\bibinfo {year} {2022})}\BibitemShut {NoStop}%
\bibitem [{\citenamefont {Taddei}\ \emph {et~al.}(2013)\citenamefont {Taddei},
  \citenamefont {Escher}, \citenamefont {Davidovich},\ and\ \citenamefont
  {de~Matos~Filho}}]{Taddei2013}%
  \BibitemOpen
  \bibfield  {author} {\bibinfo {author} {\bibfnamefont {M.~M.}\ \bibnamefont
  {Taddei}}, \bibinfo {author} {\bibfnamefont {B.~M.}\ \bibnamefont {Escher}},
  \bibinfo {author} {\bibfnamefont {L.}~\bibnamefont {Davidovich}}, \ and\
  \bibinfo {author} {\bibfnamefont {R.~L.}\ \bibnamefont {de~Matos~Filho}},\
  }\bibfield  {title} {\enquote {\bibinfo {title} {Quantum speed limit for
  physical processes},}\ }\href {\doibase 10.1103/PhysRevLett.110.050402}
  {\bibfield  {journal} {\bibinfo  {journal} {Physical Review Letters}\
  }\textbf {\bibinfo {volume} {110}},\ \bibinfo {pages} {050402} (\bibinfo
  {year} {2013})}\BibitemShut {NoStop}%
\bibitem [{\citenamefont {Braunstein}\ and\ \citenamefont
  {Caves}(1994)}]{Braunstein1994}%
  \BibitemOpen
  \bibfield  {author} {\bibinfo {author} {\bibfnamefont {Samuel~L.}\
  \bibnamefont {Braunstein}}\ and\ \bibinfo {author} {\bibfnamefont
  {Carlton~M.}\ \bibnamefont {Caves}},\ }\bibfield  {title} {\enquote {\bibinfo
  {title} {Statistical distance and the geometry of quantum states},}\ }\href
  {\doibase 10.1103/PhysRevLett.72.3439} {\bibfield  {journal} {\bibinfo
  {journal} {Physical Review Letters}\ }\textbf {\bibinfo {volume} {72}},\
  \bibinfo {pages} {3439--3443} (\bibinfo {year} {1994})}\BibitemShut {NoStop}%
\bibitem [{\citenamefont {Bagchi}\ \emph {et~al.}(2022)\citenamefont {Bagchi},
  \citenamefont {Srivastav},\ and\ \citenamefont {Pati}}]{Bagchi2022}%
  \BibitemOpen
  \bibfield  {author} {\bibinfo {author} {\bibfnamefont {Shrobona}\
  \bibnamefont {Bagchi}}, \bibinfo {author} {\bibfnamefont {Abhay}\
  \bibnamefont {Srivastav}}, \ and\ \bibinfo {author} {\bibfnamefont
  {Arun~Kumar}\ \bibnamefont {Pati}},\ }\bibfield  {title} {\enquote {\bibinfo
  {title} {Quantum speed limit from tighter uncertainty relation},}\ }\href
  {https://doi.org/10.48550/arXiv.2211.14561} {\bibfield  {journal} {\bibinfo
  {journal} {arXiv preprint arXiv:2211.14561}\ } (\bibinfo {year}
  {2022})}\BibitemShut {NoStop}%
\bibitem [{\citenamefont {Bukov}\ \emph {et~al.}(2019)\citenamefont {Bukov},
  \citenamefont {Sels},\ and\ \citenamefont {Polkovnikov}}]{Bukov2019}%
  \BibitemOpen
  \bibfield  {author} {\bibinfo {author} {\bibfnamefont {Marin}\ \bibnamefont
  {Bukov}}, \bibinfo {author} {\bibfnamefont {Dries}\ \bibnamefont {Sels}}, \
  and\ \bibinfo {author} {\bibfnamefont {Anatoli}\ \bibnamefont
  {Polkovnikov}},\ }\bibfield  {title} {\enquote {\bibinfo {title} {Geometric
  speed limit of accessible many-body state preparation},}\ }\href {\doibase
  10.1103/PhysRevX.9.011034} {\bibfield  {journal} {\bibinfo  {journal}
  {Physical Review X}\ }\textbf {\bibinfo {volume} {9}},\ \bibinfo {pages}
  {011034} (\bibinfo {year} {2019})}\BibitemShut {NoStop}%
\bibitem [{\citenamefont {Brouzos}\ \emph {et~al.}(2015)\citenamefont
  {Brouzos}, \citenamefont {Streltsov}, \citenamefont {Negretti}, \citenamefont
  {Said}, \citenamefont {Caneva}, \citenamefont {Montangero},\ and\
  \citenamefont {Calarco}}]{Brouzos2015}%
  \BibitemOpen
  \bibfield  {author} {\bibinfo {author} {\bibfnamefont {Ioannis}\ \bibnamefont
  {Brouzos}}, \bibinfo {author} {\bibfnamefont {Alexej~I.}\ \bibnamefont
  {Streltsov}}, \bibinfo {author} {\bibfnamefont {Antonio}\ \bibnamefont
  {Negretti}}, \bibinfo {author} {\bibfnamefont {Ressa~S.}\ \bibnamefont
  {Said}}, \bibinfo {author} {\bibfnamefont {Tommaso}\ \bibnamefont {Caneva}},
  \bibinfo {author} {\bibfnamefont {Simone}\ \bibnamefont {Montangero}}, \ and\
  \bibinfo {author} {\bibfnamefont {Tommaso}\ \bibnamefont {Calarco}},\
  }\bibfield  {title} {\enquote {\bibinfo {title} {Quantum speed limit and
  optimal control of many-boson dynamics},}\ }\href {\doibase
  10.1103/PhysRevA.92.062110} {\bibfield  {journal} {\bibinfo  {journal}
  {Physical Review A}\ }\textbf {\bibinfo {volume} {92}},\ \bibinfo {pages}
  {062110} (\bibinfo {year} {2015})}\BibitemShut {NoStop}%
\bibitem [{\citenamefont {Impens}\ \emph {et~al.}(2021)\citenamefont {Impens},
  \citenamefont {D'Angelis}, \citenamefont {Pinheiro},\ and\ \citenamefont
  {Gu\'ery-Odelin}}]{Impens2021}%
  \BibitemOpen
  \bibfield  {author} {\bibinfo {author} {\bibfnamefont {F.}~\bibnamefont
  {Impens}}, \bibinfo {author} {\bibfnamefont {F.~M.}\ \bibnamefont
  {D'Angelis}}, \bibinfo {author} {\bibfnamefont {F.~A.}\ \bibnamefont
  {Pinheiro}}, \ and\ \bibinfo {author} {\bibfnamefont {D.}~\bibnamefont
  {Gu\'ery-Odelin}},\ }\bibfield  {title} {\enquote {\bibinfo {title} {Time
  scaling and quantum speed limit in non-hermitian hamiltonians},}\ }\href
  {\doibase 10.1103/PhysRevA.104.052620} {\bibfield  {journal} {\bibinfo
  {journal} {Physical Review A}\ }\textbf {\bibinfo {volume} {104}},\ \bibinfo
  {pages} {052620} (\bibinfo {year} {2021})}\BibitemShut {NoStop}%
\bibitem [{\citenamefont {Uzdin}\ \emph {et~al.}(2012)\citenamefont {Uzdin},
  \citenamefont {Günther}, \citenamefont {Rahav},\ and\ \citenamefont
  {Moiseyev}}]{Uzdin2012}%
  \BibitemOpen
  \bibfield  {author} {\bibinfo {author} {\bibfnamefont {Raam}\ \bibnamefont
  {Uzdin}}, \bibinfo {author} {\bibfnamefont {Uwe}\ \bibnamefont {Günther}},
  \bibinfo {author} {\bibfnamefont {Saar}\ \bibnamefont {Rahav}}, \ and\
  \bibinfo {author} {\bibfnamefont {Nimrod}\ \bibnamefont {Moiseyev}},\
  }\bibfield  {title} {\enquote {\bibinfo {title} {Time-dependent hamiltonians
  with 100\% evolution speed efficiency},}\ }\href {\doibase
  10.1088/1751-8113/45/41/415304} {\bibfield  {journal} {\bibinfo  {journal}
  {Journal of Physics A: Mathematical and Theoretical}\ }\textbf {\bibinfo
  {volume} {45}},\ \bibinfo {pages} {415304} (\bibinfo {year}
  {2012})}\BibitemShut {NoStop}%
\bibitem [{\citenamefont {Alipour}\ \emph {et~al.}(2020)\citenamefont
  {Alipour}, \citenamefont {Chenu}, \citenamefont {Rezakhani},\ and\
  \citenamefont {del Campo}}]{Alipour2020}%
  \BibitemOpen
  \bibfield  {author} {\bibinfo {author} {\bibfnamefont {Sahar}\ \bibnamefont
  {Alipour}}, \bibinfo {author} {\bibfnamefont {Aurelia}\ \bibnamefont
  {Chenu}}, \bibinfo {author} {\bibfnamefont {Ali~T.}\ \bibnamefont
  {Rezakhani}}, \ and\ \bibinfo {author} {\bibfnamefont {Adolfo}\ \bibnamefont
  {del Campo}},\ }\bibfield  {title} {\enquote {\bibinfo {title} {Shortcuts to
  {A}diabaticity in {D}riven {O}pen {Q}uantum {S}ystems: {B}alanced {G}ain and
  {L}oss and {N}on-{M}arkovian {E}volution},}\ }\href {\doibase
  10.22331/q-2020-09-28-336} {\bibfield  {journal} {\bibinfo  {journal}
  {{Quantum}}\ }\textbf {\bibinfo {volume} {4}},\ \bibinfo {pages} {336}
  (\bibinfo {year} {2020})}\BibitemShut {NoStop}%
\bibitem [{\citenamefont {Thakuria}\ \emph {et~al.}(2022)\citenamefont
  {Thakuria}, \citenamefont {Srivastav}, \citenamefont {Mohan}, \citenamefont
  {Kumari},\ and\ \citenamefont {Pati}}]{Dimpi2022}%
  \BibitemOpen
  \bibfield  {author} {\bibinfo {author} {\bibfnamefont {Dimpi}\ \bibnamefont
  {Thakuria}}, \bibinfo {author} {\bibfnamefont {Abhay}\ \bibnamefont
  {Srivastav}}, \bibinfo {author} {\bibfnamefont {Brij}\ \bibnamefont {Mohan}},
  \bibinfo {author} {\bibfnamefont {Asmita}\ \bibnamefont {Kumari}}, \ and\
  \bibinfo {author} {\bibfnamefont {Arun~Kumar}\ \bibnamefont {Pati}},\
  }\bibfield  {title} {\enquote {\bibinfo {title} {Generalised quantum speed
  limit for arbitrary evolution},}\ }\href
  {https://doi.org/10.48550/arXiv.2207.04124} {\bibfield  {journal} {\bibinfo
  {journal} {arXiv preprint arXiv:2207.04124}\ } (\bibinfo {year}
  {2022})}\BibitemShut {NoStop}%
\bibitem [{\citenamefont {Lan}\ \emph {et~al.}(2022)\citenamefont {Lan},
  \citenamefont {Xie},\ and\ \citenamefont {Cai}}]{Lan2022}%
  \BibitemOpen
  \bibfield  {author} {\bibinfo {author} {\bibfnamefont {Kang}\ \bibnamefont
  {Lan}}, \bibinfo {author} {\bibfnamefont {Shijie}\ \bibnamefont {Xie}}, \
  and\ \bibinfo {author} {\bibfnamefont {Xiangji}\ \bibnamefont {Cai}},\
  }\bibfield  {title} {\enquote {\bibinfo {title} {Geometric quantum speed
  limits for markovian dynamics in open quantum systems},}\ }\href {\doibase
  10.1088/1367-2630/ac696b} {\bibfield  {journal} {\bibinfo  {journal} {New
  Journal of Physics}\ }\textbf {\bibinfo {volume} {24}},\ \bibinfo {pages}
  {055003} (\bibinfo {year} {2022})}\BibitemShut {NoStop}%
\bibitem [{\citenamefont {Nakajima}\ and\ \citenamefont
  {Utsumi}(2022)}]{Nakajima2022}%
  \BibitemOpen
  \bibfield  {author} {\bibinfo {author} {\bibfnamefont {Satoshi}\ \bibnamefont
  {Nakajima}}\ and\ \bibinfo {author} {\bibfnamefont {Yasuhiro}\ \bibnamefont
  {Utsumi}},\ }\bibfield  {title} {\enquote {\bibinfo {title} {Speed limits of
  the trace distance for open quantum system},}\ }\href {\doibase
  10.1088/1367-2630/ac8eca} {\bibfield  {journal} {\bibinfo  {journal} {New
  Journal of Physics}\ }\textbf {\bibinfo {volume} {24}},\ \bibinfo {pages}
  {095004} (\bibinfo {year} {2022})}\BibitemShut {NoStop}%
\bibitem [{\citenamefont {Pires}\ \emph {et~al.}(2016)\citenamefont {Pires},
  \citenamefont {Cianciaruso}, \citenamefont {C\'eleri}, \citenamefont
  {Adesso},\ and\ \citenamefont {Soares-Pinto}}]{Pires2016}%
  \BibitemOpen
  \bibfield  {author} {\bibinfo {author} {\bibfnamefont {Diego~Paiva}\
  \bibnamefont {Pires}}, \bibinfo {author} {\bibfnamefont {Marco}\ \bibnamefont
  {Cianciaruso}}, \bibinfo {author} {\bibfnamefont {Lucas~C.}\ \bibnamefont
  {C\'eleri}}, \bibinfo {author} {\bibfnamefont {Gerardo}\ \bibnamefont
  {Adesso}}, \ and\ \bibinfo {author} {\bibfnamefont {Diogo~O.}\ \bibnamefont
  {Soares-Pinto}},\ }\bibfield  {title} {\enquote {\bibinfo {title}
  {Generalized geometric quantum speed limits},}\ }\href {\doibase
  10.1103/PhysRevX.6.021031} {\bibfield  {journal} {\bibinfo  {journal}
  {Physical Review X}\ }\textbf {\bibinfo {volume} {6}},\ \bibinfo {pages}
  {021031} (\bibinfo {year} {2016})}\BibitemShut {NoStop}%
\bibitem [{\citenamefont {Brody}\ and\ \citenamefont
  {Longstaff}(2019)}]{Brody2019}%
  \BibitemOpen
  \bibfield  {author} {\bibinfo {author} {\bibfnamefont {Dorje~C.}\
  \bibnamefont {Brody}}\ and\ \bibinfo {author} {\bibfnamefont {Bradley}\
  \bibnamefont {Longstaff}},\ }\bibfield  {title} {\enquote {\bibinfo {title}
  {Evolution speed of open quantum dynamics},}\ }\href {\doibase
  10.1103/PhysRevResearch.1.033127} {\bibfield  {journal} {\bibinfo  {journal}
  {Physical Review Research}\ }\textbf {\bibinfo {volume} {1}},\ \bibinfo
  {pages} {033127} (\bibinfo {year} {2019})}\BibitemShut {NoStop}%
\bibitem [{\citenamefont {Funo}\ \emph {et~al.}(2019)\citenamefont {Funo},
  \citenamefont {Shiraishi},\ and\ \citenamefont {Saito}}]{Funo2019}%
  \BibitemOpen
  \bibfield  {author} {\bibinfo {author} {\bibfnamefont {Ken}\ \bibnamefont
  {Funo}}, \bibinfo {author} {\bibfnamefont {Naoto}\ \bibnamefont {Shiraishi}},
  \ and\ \bibinfo {author} {\bibfnamefont {Keiji}\ \bibnamefont {Saito}},\
  }\bibfield  {title} {\enquote {\bibinfo {title} {Speed limit for open quantum
  systems},}\ }\href {\doibase 10.1088/1367-2630/aaf9f5} {\bibfield  {journal}
  {\bibinfo  {journal} {New Journal of Physics}\ }\textbf {\bibinfo {volume}
  {21}},\ \bibinfo {pages} {013006} (\bibinfo {year} {2019})}\BibitemShut
  {NoStop}%
\bibitem [{\citenamefont {Das}\ \emph {et~al.}(2021)\citenamefont {Das},
  \citenamefont {Bera}, \citenamefont {Chakraborty},\ and\ \citenamefont
  {Chru\ifmmode \acute{s}\else \'{s}\fi{}ci\ifmmode~\acute{n}\else
  \'{n}\fi{}ski}}]{Das2021}%
  \BibitemOpen
  \bibfield  {author} {\bibinfo {author} {\bibfnamefont {Arpan}\ \bibnamefont
  {Das}}, \bibinfo {author} {\bibfnamefont {Anindita}\ \bibnamefont {Bera}},
  \bibinfo {author} {\bibfnamefont {Sagnik}\ \bibnamefont {Chakraborty}}, \
  and\ \bibinfo {author} {\bibfnamefont {Dariusz}\ \bibnamefont {Chru\ifmmode
  \acute{s}\else \'{s}\fi{}ci\ifmmode~\acute{n}\else \'{n}\fi{}ski}},\
  }\bibfield  {title} {\enquote {\bibinfo {title} {Thermodynamics and the
  quantum speed limit in the non-markovian regime},}\ }\href {\doibase
  10.1103/PhysRevA.104.042202} {\bibfield  {journal} {\bibinfo  {journal}
  {Physical Review A}\ }\textbf {\bibinfo {volume} {104}},\ \bibinfo {pages}
  {042202} (\bibinfo {year} {2021})}\BibitemShut {NoStop}%
\bibitem [{\citenamefont {del Campo}\ \emph {et~al.}(2013)\citenamefont {del
  Campo}, \citenamefont {Egusquiza}, \citenamefont {Plenio},\ and\
  \citenamefont {Huelga}}]{Campo2013}%
  \BibitemOpen
  \bibfield  {author} {\bibinfo {author} {\bibfnamefont {A.}~\bibnamefont {del
  Campo}}, \bibinfo {author} {\bibfnamefont {I.~L.}\ \bibnamefont {Egusquiza}},
  \bibinfo {author} {\bibfnamefont {M.~B.}\ \bibnamefont {Plenio}}, \ and\
  \bibinfo {author} {\bibfnamefont {S.~F.}\ \bibnamefont {Huelga}},\ }\bibfield
   {title} {\enquote {\bibinfo {title} {Quantum speed limits in open system
  dynamics},}\ }\href {\doibase 10.1103/PhysRevLett.110.050403} {\bibfield
  {journal} {\bibinfo  {journal} {Physical Review Letters}\ }\textbf {\bibinfo
  {volume} {110}},\ \bibinfo {pages} {050403} (\bibinfo {year}
  {2013})}\BibitemShut {NoStop}%
\bibitem [{\citenamefont {Deffner}\ and\ \citenamefont
  {Lutz}(2013)}]{Deffner2013}%
  \BibitemOpen
  \bibfield  {author} {\bibinfo {author} {\bibfnamefont {Sebastian}\
  \bibnamefont {Deffner}}\ and\ \bibinfo {author} {\bibfnamefont {Eric}\
  \bibnamefont {Lutz}},\ }\bibfield  {title} {\enquote {\bibinfo {title}
  {Quantum speed limit for non-markovian dynamics},}\ }\href {\doibase
  10.1103/PhysRevLett.111.010402} {\bibfield  {journal} {\bibinfo  {journal}
  {Physical Review Letters}\ }\textbf {\bibinfo {volume} {111}},\ \bibinfo
  {pages} {010402} (\bibinfo {year} {2013})}\BibitemShut {NoStop}%
\bibitem [{\citenamefont {Mohan}\ and\ \citenamefont
  {Pati}(2022)}]{B.Mohan2022}%
  \BibitemOpen
  \bibfield  {author} {\bibinfo {author} {\bibfnamefont {Brij}\ \bibnamefont
  {Mohan}}\ and\ \bibinfo {author} {\bibfnamefont {Arun~Kumar}\ \bibnamefont
  {Pati}},\ }\bibfield  {title} {\enquote {\bibinfo {title} {Quantum speed
  limits for observables},}\ }\href {\doibase 10.1103/PhysRevA.106.042436}
  {\bibfield  {journal} {\bibinfo  {journal} {Physical Review A}\ }\textbf
  {\bibinfo {volume} {106}},\ \bibinfo {pages} {042436} (\bibinfo {year}
  {2022})}\BibitemShut {NoStop}%
\bibitem [{\citenamefont {Garc\'{\i}a-Pintos}\ \emph
  {et~al.}(2022)\citenamefont {Garc\'{\i}a-Pintos}, \citenamefont {Nicholson},
  \citenamefont {Green}, \citenamefont {del Campo},\ and\ \citenamefont
  {Gorshkov}}]{Pintos2022}%
  \BibitemOpen
  \bibfield  {author} {\bibinfo {author} {\bibfnamefont {Luis~Pedro}\
  \bibnamefont {Garc\'{\i}a-Pintos}}, \bibinfo {author} {\bibfnamefont
  {Schuyler~B.}\ \bibnamefont {Nicholson}}, \bibinfo {author} {\bibfnamefont
  {Jason~R.}\ \bibnamefont {Green}}, \bibinfo {author} {\bibfnamefont {Adolfo}\
  \bibnamefont {del Campo}}, \ and\ \bibinfo {author} {\bibfnamefont
  {Alexey~V.}\ \bibnamefont {Gorshkov}},\ }\bibfield  {title} {\enquote
  {\bibinfo {title} {Unifying quantum and classical speed limits on
  observables},}\ }\href {\doibase 10.1103/PhysRevX.12.011038} {\bibfield
  {journal} {\bibinfo  {journal} {Physical Review X}\ }\textbf {\bibinfo
  {volume} {12}},\ \bibinfo {pages} {011038} (\bibinfo {year}
  {2022})}\BibitemShut {NoStop}%
\bibitem [{\citenamefont {Hamazaki}(2022)}]{Hamazaki2022}%
  \BibitemOpen
  \bibfield  {author} {\bibinfo {author} {\bibfnamefont {Ryusuke}\ \bibnamefont
  {Hamazaki}},\ }\bibfield  {title} {\enquote {\bibinfo {title} {Speed limits
  for macroscopic transitions},}\ }\href {\doibase 10.1103/PRXQuantum.3.020319}
  {\bibfield  {journal} {\bibinfo  {journal} {PRX Quantum}\ }\textbf {\bibinfo
  {volume} {3}},\ \bibinfo {pages} {020319} (\bibinfo {year}
  {2022})}\BibitemShut {NoStop}%
\bibitem [{\citenamefont {Aifer}\ and\ \citenamefont
  {Deffner}(2022)}]{Aifer2022}%
  \BibitemOpen
  \bibfield  {author} {\bibinfo {author} {\bibfnamefont {Maxwell}\ \bibnamefont
  {Aifer}}\ and\ \bibinfo {author} {\bibfnamefont {Sebastian}\ \bibnamefont
  {Deffner}},\ }\bibfield  {title} {\enquote {\bibinfo {title} {From quantum
  speed limits to energy-efficient quantum gates},}\ }\href {\doibase
  10.1088/1367-2630/ac6821} {\bibfield  {journal} {\bibinfo  {journal} {New
  Journal of Physics}\ }\textbf {\bibinfo {volume} {24}},\ \bibinfo {pages}
  {055002} (\bibinfo {year} {2022})}\BibitemShut {NoStop}%
\bibitem [{\citenamefont {Zwierz}\ \emph {et~al.}(2010)\citenamefont {Zwierz},
  \citenamefont {P\'erez-Delgado},\ and\ \citenamefont {Kok}}]{Zwierz2010}%
  \BibitemOpen
  \bibfield  {author} {\bibinfo {author} {\bibfnamefont {Marcin}\ \bibnamefont
  {Zwierz}}, \bibinfo {author} {\bibfnamefont {Carlos~A.}\ \bibnamefont
  {P\'erez-Delgado}}, \ and\ \bibinfo {author} {\bibfnamefont {Pieter}\
  \bibnamefont {Kok}},\ }\bibfield  {title} {\enquote {\bibinfo {title}
  {General optimality of the heisenberg limit for quantum metrology},}\ }\href
  {\doibase 10.1103/PhysRevLett.105.180402} {\bibfield  {journal} {\bibinfo
  {journal} {Physical Review Letters}\ }\textbf {\bibinfo {volume} {105}},\
  \bibinfo {pages} {180402} (\bibinfo {year} {2010})}\BibitemShut {NoStop}%
\bibitem [{\citenamefont {Campbell}\ \emph {et~al.}(2018)\citenamefont
  {Campbell}, \citenamefont {Genoni},\ and\ \citenamefont
  {Deffner}}]{Campbell2018}%
  \BibitemOpen
  \bibfield  {author} {\bibinfo {author} {\bibfnamefont {Steve}\ \bibnamefont
  {Campbell}}, \bibinfo {author} {\bibfnamefont {Marco~G}\ \bibnamefont
  {Genoni}}, \ and\ \bibinfo {author} {\bibfnamefont {Sebastian}\ \bibnamefont
  {Deffner}},\ }\bibfield  {title} {\enquote {\bibinfo {title} {Precision
  thermometry and the quantum speed limit},}\ }\href {\doibase
  10.1088/2058-9565/aaa641} {\bibfield  {journal} {\bibinfo  {journal} {Quantum
  Science and Technology}\ }\textbf {\bibinfo {volume} {3}},\ \bibinfo {pages}
  {025002} (\bibinfo {year} {2018})}\BibitemShut {NoStop}%
\bibitem [{\citenamefont {Caneva}\ \emph {et~al.}(2009)\citenamefont {Caneva},
  \citenamefont {Murphy}, \citenamefont {Calarco}, \citenamefont {Fazio},
  \citenamefont {Montangero}, \citenamefont {Giovannetti},\ and\ \citenamefont
  {Santoro}}]{Caneva2009}%
  \BibitemOpen
  \bibfield  {author} {\bibinfo {author} {\bibfnamefont {T.}~\bibnamefont
  {Caneva}}, \bibinfo {author} {\bibfnamefont {M.}~\bibnamefont {Murphy}},
  \bibinfo {author} {\bibfnamefont {T.}~\bibnamefont {Calarco}}, \bibinfo
  {author} {\bibfnamefont {R.}~\bibnamefont {Fazio}}, \bibinfo {author}
  {\bibfnamefont {S.}~\bibnamefont {Montangero}}, \bibinfo {author}
  {\bibfnamefont {V.}~\bibnamefont {Giovannetti}}, \ and\ \bibinfo {author}
  {\bibfnamefont {G.~E.}\ \bibnamefont {Santoro}},\ }\bibfield  {title}
  {\enquote {\bibinfo {title} {Optimal control at the quantum speed limit},}\
  }\href {\doibase 10.1103/PhysRevLett.103.240501} {\bibfield  {journal}
  {\bibinfo  {journal} {Physical Review Letters}\ }\textbf {\bibinfo {volume}
  {103}},\ \bibinfo {pages} {240501} (\bibinfo {year} {2009})}\BibitemShut
  {NoStop}%
\bibitem [{\citenamefont {Murphy}\ \emph {et~al.}(2010)\citenamefont {Murphy},
  \citenamefont {Montangero}, \citenamefont {Giovannetti},\ and\ \citenamefont
  {Calarco}}]{Murphy2010}%
  \BibitemOpen
  \bibfield  {author} {\bibinfo {author} {\bibfnamefont {Michael}\ \bibnamefont
  {Murphy}}, \bibinfo {author} {\bibfnamefont {Simone}\ \bibnamefont
  {Montangero}}, \bibinfo {author} {\bibfnamefont {Vittorio}\ \bibnamefont
  {Giovannetti}}, \ and\ \bibinfo {author} {\bibfnamefont {Tommaso}\
  \bibnamefont {Calarco}},\ }\bibfield  {title} {\enquote {\bibinfo {title}
  {Communication at the quantum speed limit along a spin chain},}\ }\href
  {\doibase 10.1103/PhysRevA.82.022318} {\bibfield  {journal} {\bibinfo
  {journal} {Physical Review A}\ }\textbf {\bibinfo {volume} {82}},\ \bibinfo
  {pages} {022318} (\bibinfo {year} {2010})}\BibitemShut {NoStop}%
\bibitem [{\citenamefont {Campaioli}\ \emph {et~al.}(2017)\citenamefont
  {Campaioli}, \citenamefont {Pollock}, \citenamefont {Binder}, \citenamefont
  {C\'eleri}, \citenamefont {Goold}, \citenamefont {Vinjanampathy},\ and\
  \citenamefont {Modi}}]{Campaioli2017}%
  \BibitemOpen
  \bibfield  {author} {\bibinfo {author} {\bibfnamefont {Francesco}\
  \bibnamefont {Campaioli}}, \bibinfo {author} {\bibfnamefont {Felix~A.}\
  \bibnamefont {Pollock}}, \bibinfo {author} {\bibfnamefont {Felix~C.}\
  \bibnamefont {Binder}}, \bibinfo {author} {\bibfnamefont {Lucas}\
  \bibnamefont {C\'eleri}}, \bibinfo {author} {\bibfnamefont {John}\
  \bibnamefont {Goold}}, \bibinfo {author} {\bibfnamefont {Sai}\ \bibnamefont
  {Vinjanampathy}}, \ and\ \bibinfo {author} {\bibfnamefont {Kavan}\
  \bibnamefont {Modi}},\ }\bibfield  {title} {\enquote {\bibinfo {title}
  {Enhancing the charging power of quantum batteries},}\ }\href {\doibase
  10.1103/PhysRevLett.118.150601} {\bibfield  {journal} {\bibinfo  {journal}
  {Physical Review Letters}\ }\textbf {\bibinfo {volume} {118}},\ \bibinfo
  {pages} {150601} (\bibinfo {year} {2017})}\BibitemShut {NoStop}%
\bibitem [{\citenamefont {Mohan}\ and\ \citenamefont {Pati}(2021)}]{Mohan2021}%
  \BibitemOpen
  \bibfield  {author} {\bibinfo {author} {\bibfnamefont {Brij}\ \bibnamefont
  {Mohan}}\ and\ \bibinfo {author} {\bibfnamefont {Arun~K.}\ \bibnamefont
  {Pati}},\ }\bibfield  {title} {\enquote {\bibinfo {title} {Reverse quantum
  speed limit: How slowly a quantum battery can discharge},}\ }\href {\doibase
  10.1103/PhysRevA.104.042209} {\bibfield  {journal} {\bibinfo  {journal}
  {Physical Review A}\ }\textbf {\bibinfo {volume} {104}},\ \bibinfo {pages}
  {042209} (\bibinfo {year} {2021})}\BibitemShut {NoStop}%
\bibitem [{\citenamefont {Gyhm}\ \emph {et~al.}(2022)\citenamefont {Gyhm},
  \citenamefont {\ifmmode~\check{S}\else \v{S}\fi{}afr\'anek},\ and\
  \citenamefont {Rosa}}]{Gyhm2022}%
  \BibitemOpen
  \bibfield  {author} {\bibinfo {author} {\bibfnamefont {Ju-Yeon}\ \bibnamefont
  {Gyhm}}, \bibinfo {author} {\bibfnamefont {Dominik}\ \bibnamefont
  {\ifmmode~\check{S}\else \v{S}\fi{}afr\'anek}}, \ and\ \bibinfo {author}
  {\bibfnamefont {Dario}\ \bibnamefont {Rosa}},\ }\bibfield  {title} {\enquote
  {\bibinfo {title} {Quantum charging advantage cannot be extensive without
  global operations},}\ }\href {\doibase 10.1103/PhysRevLett.128.140501}
  {\bibfield  {journal} {\bibinfo  {journal} {Physical Review Letters}\
  }\textbf {\bibinfo {volume} {128}},\ \bibinfo {pages} {140501} (\bibinfo
  {year} {2022})}\BibitemShut {NoStop}%
\bibitem [{\citenamefont {Carabba}\ \emph {et~al.}(2022)\citenamefont
  {Carabba}, \citenamefont {H{\"{o}}rnedal},\ and\ \citenamefont
  {Campo}}]{Carabba2022}%
  \BibitemOpen
  \bibfield  {author} {\bibinfo {author} {\bibfnamefont {Nicoletta}\
  \bibnamefont {Carabba}}, \bibinfo {author} {\bibfnamefont {Niklas}\
  \bibnamefont {H{\"{o}}rnedal}}, \ and\ \bibinfo {author} {\bibfnamefont
  {Adolfo~del}\ \bibnamefont {Campo}},\ }\bibfield  {title} {\enquote {\bibinfo
  {title} {Quantum speed limits on operator flows and correlation functions},}\
  }\href {\doibase 10.22331/q-2022-12-22-884} {\bibfield  {journal} {\bibinfo
  {journal} {{Quantum}}\ }\textbf {\bibinfo {volume} {6}},\ \bibinfo {pages}
  {884} (\bibinfo {year} {2022})}\BibitemShut {NoStop}%
\bibitem [{\citenamefont {Jing}\ \emph {et~al.}(2016)\citenamefont {Jing},
  \citenamefont {Wu},\ and\ \citenamefont {del Campo}}]{Jing2016}%
  \BibitemOpen
  \bibfield  {author} {\bibinfo {author} {\bibfnamefont {Jun}\ \bibnamefont
  {Jing}}, \bibinfo {author} {\bibfnamefont {Lian-Ao}\ \bibnamefont {Wu}}, \
  and\ \bibinfo {author} {\bibfnamefont {Adolfo}\ \bibnamefont {del Campo}},\
  }\bibfield  {title} {\enquote {\bibinfo {title} {Fundamental speed limits to
  the generation of quantumness},}\ }\href {\doibase 10.1038/srep38149}
  {\bibfield  {journal} {\bibinfo  {journal} {Scientific Reports}\ }\textbf
  {\bibinfo {volume} {6}},\ \bibinfo {pages} {38149} (\bibinfo {year}
  {2016})}\BibitemShut {NoStop}%
\bibitem [{\citenamefont {Pandey}\ \emph
  {et~al.}(2023{\natexlab{a}})\citenamefont {Pandey}, \citenamefont {Shrimali},
  \citenamefont {Mohan}, \citenamefont {Das},\ and\ \citenamefont
  {Pati}}]{Pandey2022}%
  \BibitemOpen
  \bibfield  {author} {\bibinfo {author} {\bibfnamefont {Vivek}\ \bibnamefont
  {Pandey}}, \bibinfo {author} {\bibfnamefont {Divyansh}\ \bibnamefont
  {Shrimali}}, \bibinfo {author} {\bibfnamefont {Brij}\ \bibnamefont {Mohan}},
  \bibinfo {author} {\bibfnamefont {Siddhartha}\ \bibnamefont {Das}}, \ and\
  \bibinfo {author} {\bibfnamefont {Arun~Kumar}\ \bibnamefont {Pati}},\
  }\bibfield  {title} {\enquote {\bibinfo {title} {Speed limits on correlations
  in bipartite quantum systems},}\ }\href {\doibase
  10.1103/PhysRevA.107.052419} {\bibfield  {journal} {\bibinfo  {journal}
  {Phys. Rev. A}\ }\textbf {\bibinfo {volume} {107}},\ \bibinfo {pages}
  {052419} (\bibinfo {year} {2023}{\natexlab{a}})}\BibitemShut {NoStop}%
\bibitem [{\citenamefont {Pandey}\ \emph
  {et~al.}(2023{\natexlab{b}})\citenamefont {Pandey}, \citenamefont {Bhowmick},
  \citenamefont {Mohan}, \citenamefont {Sohail},\ and\ \citenamefont
  {Sen}}]{Pandey2023}%
  \BibitemOpen
  \bibfield  {author} {\bibinfo {author} {\bibfnamefont {Vivek}\ \bibnamefont
  {Pandey}}, \bibinfo {author} {\bibfnamefont {Swapnil}\ \bibnamefont
  {Bhowmick}}, \bibinfo {author} {\bibfnamefont {Brij}\ \bibnamefont {Mohan}},
  \bibinfo {author} {\bibnamefont {Sohail}}, \ and\ \bibinfo {author}
  {\bibfnamefont {Ujjwal}\ \bibnamefont {Sen}},\ }\bibfield  {title} {\enquote
  {\bibinfo {title} {Fundamental speed limits on entanglement dynamics of
  bipartite quantum systems},}\ }\href
  {https://doi.org/10.48550/arXiv.2303.07415} {\bibfield  {journal} {\bibinfo
  {journal} {arXiv preprint arXiv:2303.07415}\ } (\bibinfo {year}
  {2023}{\natexlab{b}})}\BibitemShut {NoStop}%
\bibitem [{\citenamefont {Mohan}\ \emph {et~al.}(2022)\citenamefont {Mohan},
  \citenamefont {Das},\ and\ \citenamefont {Pati}}]{Mohan2022}%
  \BibitemOpen
  \bibfield  {author} {\bibinfo {author} {\bibfnamefont {Brij}\ \bibnamefont
  {Mohan}}, \bibinfo {author} {\bibfnamefont {Siddhartha}\ \bibnamefont {Das}},
  \ and\ \bibinfo {author} {\bibfnamefont {Arun~Kumar}\ \bibnamefont {Pati}},\
  }\bibfield  {title} {\enquote {\bibinfo {title} {Quantum speed limits for
  information and coherence},}\ }\href {\doibase 10.1088/1367-2630/ac753c}
  {\bibfield  {journal} {\bibinfo  {journal} {New Journal of Physics}\ }\textbf
  {\bibinfo {volume} {24}},\ \bibinfo {pages} {065003} (\bibinfo {year}
  {2022})}\BibitemShut {NoStop}%
\bibitem [{\citenamefont {Vikram}\ and\ \citenamefont
  {Galitski}(2022)}]{vikram2022}%
  \BibitemOpen
  \bibfield  {author} {\bibinfo {author} {\bibfnamefont {Amit}\ \bibnamefont
  {Vikram}}\ and\ \bibinfo {author} {\bibfnamefont {Victor}\ \bibnamefont
  {Galitski}},\ }\bibfield  {title} {\enquote {\bibinfo {title} {Exact
  universal bounds on quantum dynamics and fast scrambling},}\ }\href
  {https://doi.org/10.48550/arXiv.2212.14021} {\bibfield  {journal} {\bibinfo
  {journal} {arXiv preprint arXiv:2212.14021}\ } (\bibinfo {year}
  {2022})}\BibitemShut {NoStop}%
\bibitem [{\citenamefont {Gong}\ and\ \citenamefont
  {Hamazaki}(2022)}]{Gong2022}%
  \BibitemOpen
  \bibfield  {author} {\bibinfo {author} {\bibfnamefont {Zongping}\
  \bibnamefont {Gong}}\ and\ \bibinfo {author} {\bibfnamefont {Ryusuke}\
  \bibnamefont {Hamazaki}},\ }\bibfield  {title} {\enquote {\bibinfo {title}
  {Bounds in nonequilibrium quantum dynamics},}\ }\href {\doibase
  10.1142/S0217979222300079} {\bibfield  {journal} {\bibinfo  {journal}
  {International Journal of Modern Physics B}\ }\textbf {\bibinfo {volume}
  {36}},\ \bibinfo {pages} {2230007} (\bibinfo {year} {2022})},\ \Eprint
  {http://arxiv.org/abs/https://doi.org/10.1142/S0217979222300079}
  {https://doi.org/10.1142/S0217979222300079} \BibitemShut {NoStop}%
\bibitem [{\citenamefont {Hamilton}\ and\ \citenamefont
  {Clark}(2023)}]{Hamilton2023}%
  \BibitemOpen
  \bibfield  {author} {\bibinfo {author} {\bibfnamefont {Gregory~A.}\
  \bibnamefont {Hamilton}}\ and\ \bibinfo {author} {\bibfnamefont {Bryan~K.}\
  \bibnamefont {Clark}},\ }\bibfield  {title} {\enquote {\bibinfo {title}
  {Quantifying unitary flow efficiency and entanglement for many-body
  localization},}\ }\href {\doibase 10.1103/PhysRevB.107.064203} {\bibfield
  {journal} {\bibinfo  {journal} {Physical Review B}\ }\textbf {\bibinfo
  {volume} {107}},\ \bibinfo {pages} {064203} (\bibinfo {year}
  {2023})}\BibitemShut {NoStop}%
\bibitem [{\citenamefont {Shrimali}\ \emph {et~al.}(2022)\citenamefont
  {Shrimali}, \citenamefont {Bhowmick}, \citenamefont {Pandey},\ and\
  \citenamefont {Pati}}]{Shri2022}%
  \BibitemOpen
  \bibfield  {author} {\bibinfo {author} {\bibfnamefont {Divyansh}\
  \bibnamefont {Shrimali}}, \bibinfo {author} {\bibfnamefont {Swapnil}\
  \bibnamefont {Bhowmick}}, \bibinfo {author} {\bibfnamefont {Vivek}\
  \bibnamefont {Pandey}}, \ and\ \bibinfo {author} {\bibfnamefont {Arun~Kumar}\
  \bibnamefont {Pati}},\ }\bibfield  {title} {\enquote {\bibinfo {title}
  {Capacity of entanglement for a nonlocal hamiltonian},}\ }\href {\doibase
  10.1103/PhysRevA.106.042419} {\bibfield  {journal} {\bibinfo  {journal}
  {Physical Review A}\ }\textbf {\bibinfo {volume} {106}},\ \bibinfo {pages}
  {042419} (\bibinfo {year} {2022})}\BibitemShut {NoStop}%
\bibitem [{\citenamefont {H{\"o}rnedal}\ \emph {et~al.}(2022)\citenamefont
  {H{\"o}rnedal}, \citenamefont {Carabba}, \citenamefont {Matsoukas-Roubeas},\
  and\ \citenamefont {del Campo}}]{Niklas2022}%
  \BibitemOpen
  \bibfield  {author} {\bibinfo {author} {\bibfnamefont {Niklas}\ \bibnamefont
  {H{\"o}rnedal}}, \bibinfo {author} {\bibfnamefont {Nicoletta}\ \bibnamefont
  {Carabba}}, \bibinfo {author} {\bibfnamefont {Apollonas~S.}\ \bibnamefont
  {Matsoukas-Roubeas}}, \ and\ \bibinfo {author} {\bibfnamefont {Adolfo}\
  \bibnamefont {del Campo}},\ }\bibfield  {title} {\enquote {\bibinfo {title}
  {Ultimate speed limits to the growth of operator complexity},}\ }\href
  {\doibase 10.1038/s42005-022-00985-1} {\bibfield  {journal} {\bibinfo
  {journal} {Communications Physics}\ }\textbf {\bibinfo {volume} {5}},\
  \bibinfo {pages} {207} (\bibinfo {year} {2022})}\BibitemShut {NoStop}%
\bibitem [{\citenamefont {Heyl}(2017)}]{Heyl2017}%
  \BibitemOpen
  \bibfield  {author} {\bibinfo {author} {\bibfnamefont {Markus}\ \bibnamefont
  {Heyl}},\ }\bibfield  {title} {\enquote {\bibinfo {title} {Quenching a
  quantum critical state by the order parameter: Dynamical quantum phase
  transitions and quantum speed limits},}\ }\href {\doibase
  10.1103/PhysRevB.95.060504} {\bibfield  {journal} {\bibinfo  {journal}
  {Physical Review B}\ }\textbf {\bibinfo {volume} {95}},\ \bibinfo {pages}
  {060504} (\bibinfo {year} {2017})}\BibitemShut {NoStop}%
\bibitem [{\citenamefont {Zhou}\ \emph {et~al.}(2021)\citenamefont {Zhou},
  \citenamefont {Zeng},\ and\ \citenamefont {Chen}}]{Zhou2021}%
  \BibitemOpen
  \bibfield  {author} {\bibinfo {author} {\bibfnamefont {Bozhen}\ \bibnamefont
  {Zhou}}, \bibinfo {author} {\bibfnamefont {Yumeng}\ \bibnamefont {Zeng}}, \
  and\ \bibinfo {author} {\bibfnamefont {Shu}\ \bibnamefont {Chen}},\
  }\bibfield  {title} {\enquote {\bibinfo {title} {Exact zeros of the loschmidt
  echo and quantum speed limit time for the dynamical quantum phase transition
  in finite-size systems},}\ }\href {\doibase 10.1103/PhysRevB.104.094311}
  {\bibfield  {journal} {\bibinfo  {journal} {Physical Review B}\ }\textbf
  {\bibinfo {volume} {104}},\ \bibinfo {pages} {094311} (\bibinfo {year}
  {2021})}\BibitemShut {NoStop}%
\bibitem [{\citenamefont {Okuyama}\ and\ \citenamefont
  {Ohzeki}(2018)}]{Okuyama2018}%
  \BibitemOpen
  \bibfield  {author} {\bibinfo {author} {\bibfnamefont {Manaka}\ \bibnamefont
  {Okuyama}}\ and\ \bibinfo {author} {\bibfnamefont {Masayuki}\ \bibnamefont
  {Ohzeki}},\ }\bibfield  {title} {\enquote {\bibinfo {title} {Quantum speed
  limit is not quantum},}\ }\href {\doibase 10.1103/PhysRevLett.120.070402}
  {\bibfield  {journal} {\bibinfo  {journal} {Physical Review Letters}\
  }\textbf {\bibinfo {volume} {120}},\ \bibinfo {pages} {070402} (\bibinfo
  {year} {2018})}\BibitemShut {NoStop}%
\bibitem [{\citenamefont {Shanahan}\ \emph {et~al.}(2018)\citenamefont
  {Shanahan}, \citenamefont {Chenu}, \citenamefont {Margolus},\ and\
  \citenamefont {del Campo}}]{Shanahan2018}%
  \BibitemOpen
  \bibfield  {author} {\bibinfo {author} {\bibfnamefont {B.}~\bibnamefont
  {Shanahan}}, \bibinfo {author} {\bibfnamefont {A.}~\bibnamefont {Chenu}},
  \bibinfo {author} {\bibfnamefont {N.}~\bibnamefont {Margolus}}, \ and\
  \bibinfo {author} {\bibfnamefont {A.}~\bibnamefont {del Campo}},\ }\bibfield
  {title} {\enquote {\bibinfo {title} {Quantum speed limits across the
  quantum-to-classical transition},}\ }\href {\doibase
  10.1103/PhysRevLett.120.070401} {\bibfield  {journal} {\bibinfo  {journal}
  {Physical Review Letters}\ }\textbf {\bibinfo {volume} {120}},\ \bibinfo
  {pages} {070401} (\bibinfo {year} {2018})}\BibitemShut {NoStop}%
\bibitem [{\citenamefont {Van~Vu}\ and\ \citenamefont
  {Saito}(2023{\natexlab{b}})}]{Tan20232}%
  \BibitemOpen
  \bibfield  {author} {\bibinfo {author} {\bibfnamefont {Tan}\ \bibnamefont
  {Van~Vu}}\ and\ \bibinfo {author} {\bibfnamefont {Keiji}\ \bibnamefont
  {Saito}},\ }\bibfield  {title} {\enquote {\bibinfo {title} {Thermodynamic
  unification of optimal transport: Thermodynamic uncertainty relation, minimum
  dissipation, and thermodynamic speed limits},}\ }\href {\doibase
  10.1103/PhysRevX.13.011013} {\bibfield  {journal} {\bibinfo  {journal}
  {Physical Review X}\ }\textbf {\bibinfo {volume} {13}},\ \bibinfo {pages}
  {011013} (\bibinfo {year} {2023}{\natexlab{b}})}\BibitemShut {NoStop}%
\bibitem [{\citenamefont {Yoshimura}\ and\ \citenamefont
  {Ito}(2021)}]{Yosh2021}%
  \BibitemOpen
  \bibfield  {author} {\bibinfo {author} {\bibfnamefont {Kohei}\ \bibnamefont
  {Yoshimura}}\ and\ \bibinfo {author} {\bibfnamefont {Sosuke}\ \bibnamefont
  {Ito}},\ }\bibfield  {title} {\enquote {\bibinfo {title} {Thermodynamic
  uncertainty relation and thermodynamic speed limit in deterministic chemical
  reaction networks},}\ }\href {\doibase 10.1103/PhysRevLett.127.160601}
  {\bibfield  {journal} {\bibinfo  {journal} {Physical Review Letters}\
  }\textbf {\bibinfo {volume} {127}},\ \bibinfo {pages} {160601} (\bibinfo
  {year} {2021})}\BibitemShut {NoStop}%
\bibitem [{\citenamefont {Aghion}\ and\ \citenamefont
  {Green}(2023)}]{Aghion2023}%
  \BibitemOpen
  \bibfield  {author} {\bibinfo {author} {\bibfnamefont {Erez}\ \bibnamefont
  {Aghion}}\ and\ \bibinfo {author} {\bibfnamefont {Jason~R}\ \bibnamefont
  {Green}},\ }\bibfield  {title} {\enquote {\bibinfo {title} {Thermodynamic
  speed limits for mechanical work},}\ }\href {\doibase
  10.1088/1751-8121/acb5d6} {\bibfield  {journal} {\bibinfo  {journal} {Journal
  of Physics A: Mathematical and Theoretical}\ }\textbf {\bibinfo {volume}
  {56}},\ \bibinfo {pages} {05LT01} (\bibinfo {year} {2023})}\BibitemShut
  {NoStop}%
\bibitem [{\citenamefont {Liegener}\ and\ \citenamefont
  {Rudnicki}(2022)}]{Liegener2022}%
  \BibitemOpen
  \bibfield  {author} {\bibinfo {author} {\bibfnamefont {Klaus}\ \bibnamefont
  {Liegener}}\ and\ \bibinfo {author} {\bibfnamefont {{\L}ukasz}\ \bibnamefont
  {Rudnicki}},\ }\bibfield  {title} {\enquote {\bibinfo {title} {Quantum speed
  limit and stability of coherent states in quantum gravity},}\ }\href
  {\doibase 10.1088/1361-6382/ac6faa} {\bibfield  {journal} {\bibinfo
  {journal} {Classical and Quantum Gravity}\ }\textbf {\bibinfo {volume}
  {39}},\ \bibinfo {pages} {12LT01} (\bibinfo {year} {2022})}\BibitemShut
  {NoStop}%
\bibitem [{\citenamefont {Shabir}\ \emph {et~al.}(2023)\citenamefont {Shabir},
  \citenamefont {Wani}, \citenamefont {Ali}, \citenamefont {Kannan},
  \citenamefont {Sheikh}, \citenamefont {Faizal}, \citenamefont {Sheikh},
  \citenamefont {Rubab},\ and\ \citenamefont {Al-Kuwari}}]{Shabir2023}%
  \BibitemOpen
  \bibfield  {author} {\bibinfo {author} {\bibfnamefont {Arshid}\ \bibnamefont
  {Shabir}}, \bibinfo {author} {\bibfnamefont {Salman~Sajad}\ \bibnamefont
  {Wani}}, \bibinfo {author} {\bibfnamefont {Raja~Nisar}\ \bibnamefont {Ali}},
  \bibinfo {author} {\bibfnamefont {S}~\bibnamefont {Kannan}}, \bibinfo
  {author} {\bibfnamefont {Aasiya}\ \bibnamefont {Sheikh}}, \bibinfo {author}
  {\bibfnamefont {Mir}\ \bibnamefont {Faizal}}, \bibinfo {author}
  {\bibfnamefont {Javid~A}\ \bibnamefont {Sheikh}}, \bibinfo {author}
  {\bibfnamefont {Seemin}\ \bibnamefont {Rubab}}, \ and\ \bibinfo {author}
  {\bibfnamefont {Saif}\ \bibnamefont {Al-Kuwari}},\ }\bibfield  {title}
  {\enquote {\bibinfo {title} {A novel application of quantum speed limit to
  string theory},}\ }\href {https://doi.org/10.48550/arXiv.2302.08325}
  {\bibfield  {journal} {\bibinfo  {journal} {arXiv preprint arXiv:2302.08325}\
  } (\bibinfo {year} {2023})}\BibitemShut {NoStop}%
\bibitem [{\citenamefont {Carlini}\ \emph {et~al.}(2006)\citenamefont
  {Carlini}, \citenamefont {Hosoya}, \citenamefont {Koike},\ and\ \citenamefont
  {Okudaira}}]{Carlini2006}%
  \BibitemOpen
  \bibfield  {author} {\bibinfo {author} {\bibfnamefont {Alberto}\ \bibnamefont
  {Carlini}}, \bibinfo {author} {\bibfnamefont {Akio}\ \bibnamefont {Hosoya}},
  \bibinfo {author} {\bibfnamefont {Tatsuhiko}\ \bibnamefont {Koike}}, \ and\
  \bibinfo {author} {\bibfnamefont {Yosuke}\ \bibnamefont {Okudaira}},\
  }\bibfield  {title} {\enquote {\bibinfo {title} {Time-optimal quantum
  evolution},}\ }\href {\doibase 10.1103/PhysRevLett.96.060503} {\bibfield
  {journal} {\bibinfo  {journal} {Physical Review Letters}\ }\textbf {\bibinfo
  {volume} {96}},\ \bibinfo {pages} {060503} (\bibinfo {year}
  {2006})}\BibitemShut {NoStop}%
\bibitem [{\citenamefont {Campaioli}(2020)}]{Campaioli2020}%
  \BibitemOpen
  \bibfield  {author} {\bibinfo {author} {\bibfnamefont {Francesco}\
  \bibnamefont {Campaioli}},\ }\bibfield  {title} {\enquote {\bibinfo {title}
  {{Tightening Time-Energy Uncertainty Relations}},}\ }\href {\doibase
  10.26180/5e4226a97b667} {\  (\bibinfo {year} {2020}),\
  10.26180/5e4226a97b667}\BibitemShut {NoStop}%
\bibitem [{\citenamefont {Thakuria}\ and\ \citenamefont
  {Pati}(2022)}]{Dimpi20222}%
  \BibitemOpen
  \bibfield  {author} {\bibinfo {author} {\bibfnamefont {Dimpi}\ \bibnamefont
  {Thakuria}}\ and\ \bibinfo {author} {\bibfnamefont {Arun~Kumar}\ \bibnamefont
  {Pati}},\ }\bibfield  {title} {\enquote {\bibinfo {title} {Stronger quantum
  speed limit},}\ }\href {https://doi.org/10.48550/arXiv.2208.05469} {\bibfield
   {journal} {\bibinfo  {journal} {arXiv preprint arXiv:2208.05469}\ }
  (\bibinfo {year} {2022})}\BibitemShut {NoStop}%
\bibitem [{\citenamefont {Maccone}\ and\ \citenamefont
  {Pati}(2014)}]{Maccone2014}%
  \BibitemOpen
  \bibfield  {author} {\bibinfo {author} {\bibfnamefont {Lorenzo}\ \bibnamefont
  {Maccone}}\ and\ \bibinfo {author} {\bibfnamefont {Arun~K.}\ \bibnamefont
  {Pati}},\ }\bibfield  {title} {\enquote {\bibinfo {title} {Stronger
  uncertainty relations for all incompatible observables},}\ }\href {\doibase
  10.1103/PhysRevLett.113.260401} {\bibfield  {journal} {\bibinfo  {journal}
  {Physical Review Letters}\ }\textbf {\bibinfo {volume} {113}},\ \bibinfo
  {pages} {260401} (\bibinfo {year} {2014})}\BibitemShut {NoStop}%
\bibitem [{\citenamefont {Hall}(2001{\natexlab{a}})}]{Hall2001}%
  \BibitemOpen
  \bibfield  {author} {\bibinfo {author} {\bibfnamefont {Michael J.~W.}\
  \bibnamefont {Hall}},\ }\bibfield  {title} {\enquote {\bibinfo {title} {Exact
  uncertainty relations},}\ }\href {\doibase 10.1103/PhysRevA.64.052103}
  {\bibfield  {journal} {\bibinfo  {journal} {Physical Review A}\ }\textbf
  {\bibinfo {volume} {64}},\ \bibinfo {pages} {052103} (\bibinfo {year}
  {2001}{\natexlab{a}})}\BibitemShut {NoStop}%
\bibitem [{\citenamefont {Hall}(2001{\natexlab{b}})}]{Hall20012}%
  \BibitemOpen
  \bibfield  {author} {\bibinfo {author} {\bibfnamefont {Michael J.~W.}\
  \bibnamefont {Hall}},\ }\bibfield  {title} {\enquote {\bibinfo {title} {Exact
  uncertainty relations: technical details},}\ }\href
  {https://doi.org/10.48550/arXiv.quant-ph/0103072} {\  (\bibinfo {year}
  {2001}{\natexlab{b}})},\ \Eprint {http://arxiv.org/abs/quant-ph/0103072}
  {arXiv:quant-ph/0103072 [quant-ph]} \BibitemShut {NoStop}%
\bibitem [{\citenamefont {Hall}(2001{\natexlab{c}})}]{Hall20013}%
  \BibitemOpen
  \bibfield  {author} {\bibinfo {author} {\bibfnamefont {Michael J.~W.}\
  \bibnamefont {Hall}},\ }\bibfield  {title} {\enquote {\bibinfo {title} {Exact
  uncertainty relations: physical significance},}\ }\href
  {https://doi.org/10.1103/PhysRevA.64.052103} {\  (\bibinfo {year}
  {2001}{\natexlab{c}})},\ \Eprint {http://arxiv.org/abs/quant-ph/0107149}
  {arXiv:quant-ph/0107149 [quant-ph]} \BibitemShut {NoStop}%
\bibitem [{\citenamefont {Miller}\ and\ \citenamefont
  {Perarnau-Llobet}(2023)}]{Harry2023}%
  \BibitemOpen
  \bibfield  {author} {\bibinfo {author} {\bibfnamefont {Harry J.~D.}\
  \bibnamefont {Miller}}\ and\ \bibinfo {author} {\bibfnamefont {Martí}\
  \bibnamefont {Perarnau-Llobet}},\ }\bibfield  {title} {\enquote {\bibinfo
  {title} {{Finite-time bounds on the probabilistic violation of the second law
  of thermodynamics}},}\ }\href {\doibase 10.21468/SciPostPhys.14.4.072}
  {\bibfield  {journal} {\bibinfo  {journal} {SciPost Physics}\ }\textbf
  {\bibinfo {volume} {14}},\ \bibinfo {pages} {072} (\bibinfo {year}
  {2023})}\BibitemShut {NoStop}%
\bibitem [{\citenamefont {Robertson}(1929)}]{Robertson1929}%
  \BibitemOpen
  \bibfield  {author} {\bibinfo {author} {\bibfnamefont {H.~P.}\ \bibnamefont
  {Robertson}},\ }\bibfield  {title} {\enquote {\bibinfo {title} {The
  uncertainty principle},}\ }\href {\doibase 10.1103/PhysRev.34.163} {\bibfield
   {journal} {\bibinfo  {journal} {Physical Review}\ }\textbf {\bibinfo
  {volume} {34}},\ \bibinfo {pages} {163--164} (\bibinfo {year}
  {1929})}\BibitemShut {NoStop}%
\bibitem [{\citenamefont {Jefferson}\ and\ \citenamefont
  {Myers}(2017)}]{Jeff2017}%
  \BibitemOpen
  \bibfield  {author} {\bibinfo {author} {\bibfnamefont {Robert~A.}\
  \bibnamefont {Jefferson}}\ and\ \bibinfo {author} {\bibfnamefont {Robert~C.}\
  \bibnamefont {Myers}},\ }\bibfield  {title} {\enquote {\bibinfo {title}
  {Circuit complexity in quantum field theory},}\ }\href {\doibase
  10.1007/JHEP10(2017)107} {\bibfield  {journal} {\bibinfo  {journal} {Journal
  of High Energy Physics}\ }\textbf {\bibinfo {volume} {2017}},\ \bibinfo
  {pages} {107} (\bibinfo {year} {2017})}\BibitemShut {NoStop}%
\bibitem [{\citenamefont {Chapman}\ \emph {et~al.}(2018)\citenamefont
  {Chapman}, \citenamefont {Heller}, \citenamefont {Marrochio},\ and\
  \citenamefont {Pastawski}}]{Chapman2018}%
  \BibitemOpen
  \bibfield  {author} {\bibinfo {author} {\bibfnamefont {Shira}\ \bibnamefont
  {Chapman}}, \bibinfo {author} {\bibfnamefont {Michal~P.}\ \bibnamefont
  {Heller}}, \bibinfo {author} {\bibfnamefont {Hugo}\ \bibnamefont
  {Marrochio}}, \ and\ \bibinfo {author} {\bibfnamefont {Fernando}\
  \bibnamefont {Pastawski}},\ }\bibfield  {title} {\enquote {\bibinfo {title}
  {Toward a definition of complexity for quantum field theory states},}\ }\href
  {\doibase 10.1103/PhysRevLett.120.121602} {\bibfield  {journal} {\bibinfo
  {journal} {Physical Review Letters}\ }\textbf {\bibinfo {volume} {120}},\
  \bibinfo {pages} {121602} (\bibinfo {year} {2018})}\BibitemShut {NoStop}%
\bibitem [{\citenamefont {Nielsen}\ \emph {et~al.}(2006)\citenamefont
  {Nielsen}, \citenamefont {Dowling}, \citenamefont {Gu},\ and\ \citenamefont
  {Doherty}}]{Mic2006}%
  \BibitemOpen
  \bibfield  {author} {\bibinfo {author} {\bibfnamefont {Michael~A.}\
  \bibnamefont {Nielsen}}, \bibinfo {author} {\bibfnamefont {Mark~R.}\
  \bibnamefont {Dowling}}, \bibinfo {author} {\bibfnamefont {Mile}\
  \bibnamefont {Gu}}, \ and\ \bibinfo {author} {\bibfnamefont {Andrew~C.}\
  \bibnamefont {Doherty}},\ }\bibfield  {title} {\enquote {\bibinfo {title}
  {Quantum computation as geometry},}\ }\href {\doibase
  10.1126/science.1121541} {\bibfield  {journal} {\bibinfo  {journal}
  {Science}\ }\textbf {\bibinfo {volume} {311}},\ \bibinfo {pages} {1133--1135}
  (\bibinfo {year} {2006})}\BibitemShut {NoStop}%
\bibitem [{\citenamefont {Caputa}\ and\ \citenamefont
  {Magan}(2019)}]{Caputa2019}%
  \BibitemOpen
  \bibfield  {author} {\bibinfo {author} {\bibfnamefont {Pawe\l{}}\
  \bibnamefont {Caputa}}\ and\ \bibinfo {author} {\bibfnamefont {Javier~M.}\
  \bibnamefont {Magan}},\ }\bibfield  {title} {\enquote {\bibinfo {title}
  {Quantum computation as gravity},}\ }\href {\doibase
  10.1103/PhysRevLett.122.231302} {\bibfield  {journal} {\bibinfo  {journal}
  {Physical Review Letters}\ }\textbf {\bibinfo {volume} {122}},\ \bibinfo
  {pages} {231302} (\bibinfo {year} {2019})}\BibitemShut {NoStop}%
\bibitem [{\citenamefont {Adhikari}\ \emph {et~al.}(2021)\citenamefont
  {Adhikari}, \citenamefont {Choudhury}, \citenamefont {Chowdhury},
  \citenamefont {Shirish},\ and\ \citenamefont {Swain}}]{Adhikari2021}%
  \BibitemOpen
  \bibfield  {author} {\bibinfo {author} {\bibfnamefont {Kiran}\ \bibnamefont
  {Adhikari}}, \bibinfo {author} {\bibfnamefont {Sayantan}\ \bibnamefont
  {Choudhury}}, \bibinfo {author} {\bibfnamefont {Satyaki}\ \bibnamefont
  {Chowdhury}}, \bibinfo {author} {\bibfnamefont {K.}~\bibnamefont {Shirish}},
  \ and\ \bibinfo {author} {\bibfnamefont {Abinash}\ \bibnamefont {Swain}},\
  }\bibfield  {title} {\enquote {\bibinfo {title} {Circuit complexity as a
  novel probe of quantum entanglement: A study with black hole gas in arbitrary
  dimensions},}\ }\href {\doibase 10.1103/PhysRevD.104.065002} {\bibfield
  {journal} {\bibinfo  {journal} {Physical Review D}\ }\textbf {\bibinfo
  {volume} {104}},\ \bibinfo {pages} {065002} (\bibinfo {year}
  {2021})}\BibitemShut {NoStop}%
\bibitem [{\citenamefont {Liu}\ \emph {et~al.}(2020)\citenamefont {Liu},
  \citenamefont {Whitsitt}, \citenamefont {Curtis}, \citenamefont {Lundgren},
  \citenamefont {Titum}, \citenamefont {Yang}, \citenamefont {Garrison},\ and\
  \citenamefont {Gorshkov}}]{Liu2020}%
  \BibitemOpen
  \bibfield  {author} {\bibinfo {author} {\bibfnamefont {Fangli}\ \bibnamefont
  {Liu}}, \bibinfo {author} {\bibfnamefont {Seth}\ \bibnamefont {Whitsitt}},
  \bibinfo {author} {\bibfnamefont {Jonathan~B.}\ \bibnamefont {Curtis}},
  \bibinfo {author} {\bibfnamefont {Rex}\ \bibnamefont {Lundgren}}, \bibinfo
  {author} {\bibfnamefont {Paraj}\ \bibnamefont {Titum}}, \bibinfo {author}
  {\bibfnamefont {Zhi-Cheng}\ \bibnamefont {Yang}}, \bibinfo {author}
  {\bibfnamefont {James~R.}\ \bibnamefont {Garrison}}, \ and\ \bibinfo {author}
  {\bibfnamefont {Alexey~V.}\ \bibnamefont {Gorshkov}},\ }\bibfield  {title}
  {\enquote {\bibinfo {title} {Circuit complexity across a topological phase
  transition},}\ }\href {\doibase 10.1103/PhysRevResearch.2.013323} {\bibfield
  {journal} {\bibinfo  {journal} {Physical Review Research}\ }\textbf {\bibinfo
  {volume} {2}},\ \bibinfo {pages} {013323} (\bibinfo {year}
  {2020})}\BibitemShut {NoStop}%
\bibitem [{\citenamefont {Brown}\ and\ \citenamefont
  {Susskind}(2019)}]{Brown2019}%
  \BibitemOpen
  \bibfield  {author} {\bibinfo {author} {\bibfnamefont {Adam~R.}\ \bibnamefont
  {Brown}}\ and\ \bibinfo {author} {\bibfnamefont {Leonard}\ \bibnamefont
  {Susskind}},\ }\bibfield  {title} {\enquote {\bibinfo {title} {Complexity
  geometry of a single qubit},}\ }\href {\doibase 10.1103/PhysRevD.100.046020}
  {\bibfield  {journal} {\bibinfo  {journal} {Physical Review D}\ }\textbf
  {\bibinfo {volume} {100}},\ \bibinfo {pages} {046020} (\bibinfo {year}
  {2019})}\BibitemShut {NoStop}%
\bibitem [{\citenamefont {Poggi}(2019)}]{Poggi2019}%
  \BibitemOpen
  \bibfield  {author} {\bibinfo {author} {\bibfnamefont {Pablo~M.}\
  \bibnamefont {Poggi}},\ }\bibfield  {title} {\enquote {\bibinfo {title}
  {Geometric quantum speed limits and short-time accessibility to unitary
  operations},}\ }\href {\doibase 10.1103/PhysRevA.99.042116} {\bibfield
  {journal} {\bibinfo  {journal} {Physical Review A}\ }\textbf {\bibinfo
  {volume} {99}},\ \bibinfo {pages} {042116} (\bibinfo {year}
  {2019})}\BibitemShut {NoStop}%
\end{thebibliography}%

\newpage
\appendix
\section{Supplementary Material}

{\it \bf Detailed Proof of Theorem 1:}
Let us consider a two-dimensional quantum system which is described by pure states, and whose time evolution is governed by the Hamiltonian $H_t$. Let us take $A=\Pi_{\psi_0}=\ketbra{\psi_0}{\psi_0}$, $B=H_{t}$ and $\rho_{t}=|\psi_{t}\rangle\langle\psi_{t}|$ (where $|\psi_{t}\rangle=U_{t}|\psi_{0}\rangle$) in the exact uncertainty relation Eq. (2) (main text). Then, using the definition of $\delta_{H_t}\Pi_{\psi_0}$ defined in Eq. (3) (main text), we obtain
    \begin{align}\label{exactQSL2d1}
        (\delta_{H_t}\Pi_{\psi_0})^{-2}& =\frac{\bra{\psi_0}\frac{i}{\hbar}[H_t,\ketbra{\psi_t}{\psi_t}]\ket{\psi_0}^2}{\braket{\psi_0}{\ketbra{\psi_t}{\psi_t}| \psi_0}}\nonumber\\
        &+\frac{\bra{\psi_0^{\perp}}\frac{i}{\hbar}[H_t,\ketbra{\psi_t}{\psi_t}]\ket{\psi_0^{\perp}}^2}{\braket{\psi_0^{\perp}}{\ketbra{\psi_t}{\psi_t}| \psi_0^{\perp}}},\tag{S1}
    \end{align}
where $\{\ket{\psi_0},\ket{\psi_0^{\perp}}\}$ is the eigenbasis of $A$ .

Now, by using the Liouville-von Neumann equation   $\frac{d\rho_t}{dt}=-\frac{i}{\hbar}[H_t,\rho_t]$ (where $\rho_{t}=|\psi_{t}\rangle\langle\psi_{t}|$), we obtain
\begin{align}\label{exactQSL2d1}
        (\delta_{H_t}\Pi_{\psi_0})^{-2}&=\frac{\bra{\psi_0}\frac{d\rho_t}{dt}\ket{\psi_0}^2}{|\hspace{-1mm}\braket{\psi_0}{\psi_t}\hspace{-1mm}|^2}+\frac{\bra{\psi_0^{\perp}}\frac{d\rho_t}{dt}\ket{\psi_0^{\perp}}^2}{|\hspace{-1mm}\braket{\psi_0^{\perp}}{\psi_t}\hspace{-1mm}|^2}\nonumber\\
        &=\frac{\left(\frac{d}{dt}|\hspace{-1mm}\braket{\psi_0}{\psi_t}\hspace{-1mm}|^2\right)^2}{|\hspace{-1mm}\braket{\psi_0}{\psi_t}\hspace{-1mm}|^2}+\frac{\left(\frac{d}{dt}|\hspace{-1mm}\braket{\psi_0^{\perp}}{\psi_t}\hspace{-1mm}|^2\right)^2}{|\hspace{-1mm}\braket{\psi_0^{\perp}}{\psi_t}\hspace{-1mm}|^2}\nonumber\\
        &=\frac{\left(\frac{d}{dt}|\hspace{-1mm}\braket{\psi_0}{\psi_t}\hspace{-1mm}|^2\right)^2}{|\hspace{-1mm}\braket{\psi_0}{\psi_t}\hspace{-1mm}|^2}+\frac{\left(\frac{d}{dt}|\hspace{-1mm}\braket{\psi_0}{\psi_t}\hspace{-1mm}|^2\right)^2}{1-|\hspace{-1mm}\braket{\psi_0}{\psi_t}\hspace{-1mm}|^2},\nonumber\\
        &=\left(\frac{d}{dt}|\hspace{-1mm}\braket{\psi_0}{\psi_t}\hspace{-1mm}|^2\right)^2\left[\frac{1}{|\hspace{-1mm}\braket{\psi_0}{\psi_t}\hspace{-1mm}|^2-|\hspace{-1mm}\braket{\psi_0}{\psi_t}\hspace{-1mm}|^4}\right],\tag{S2}
    \end{align}
where we have used the fact that $\bra{\psi_0}\frac{d\rho_t}{dt}\ket{\psi_0}=\frac{d}{dt}|\hspace{-1mm}\braket{\psi_0}{\psi_t}\hspace{-1mm}|^2$ and $|\hspace{-1mm}\braket{\psi_0^{\perp}}{\psi_t}\hspace{-1mm}|^2=1-|\hspace{-1mm}\braket{\psi_0}{\psi_t}\hspace{-1mm}|^2$.\\
Next, utilizing the exact uncertainty relation (Eq. (2) of the main text), we obtain
\begin{align}
    \Delta H_t^{\text{nc}} &= \frac{\hbar}{2}(\delta_{H_t}\Pi_{\psi_0})^{-1}\nonumber\\
    &=\pm\frac{\hbar}{2}\frac{\frac{d}{dt}|\hspace{-1mm}\braket{\psi_0}{\psi_t}\hspace{-1mm}|^2}{\sqrt{|\hspace{-1mm}\braket{\psi_0}{\psi_t}\hspace{-1mm}|^2-|\hspace{-1mm}\braket{\psi_0}{\psi_t}\hspace{-1mm}|^4}}.\tag{S3}
\end{align}
Since $\Delta H_t^{\text{nc}}$ is always positive in Eq.~(S3),
therefore we assume that $p_t=|\hspace{-1mm}\braket{\psi_0}{\psi_t}\hspace{-1mm}|^2$ is monotonically decreasing and hence
\begin{align}\label{exactQSL2d2}
    \Delta H_t^{\text{nc}} = -\frac{\hbar}{2}\frac{\frac{d}{dt} p_t}{\sqrt{p_t-p_t^2}}.\tag{S4}
\end{align}
After integrating both sides of Eq.~(S4) with respect to time,  we obtain 
  \begin{align}\label{exactQSL2d2}
    T\left(\frac{1}{T}\int_{0}^{T}\Delta H_t^{\text{nc}}dt\right) &= -\frac{\hbar}{2}\int_{0}^{T}\frac{\frac{d}{dt} p_t}{\sqrt{p_t-p_t^2}}dt.\tag{S6}
\end{align}
Now defining $\langle\!\langle \Delta H^{\text{\rm nc}}_t \rangle\!\rangle_{T}\equiv\frac{1}{T}\!\int_{0}^{T}\!\!\Delta H_t^{\text{nc}}\,dt$ as the time-average uncertainty of the non-classical part of the driving Hamiltonian in state $|\psi_{t}\rangle$, we obtain
\begin{align}\label{exactQSL2d2}
    T &= -\frac{\hbar}{2\langle\!\langle \Delta H^{\text{\rm nc}}_t \rangle\!\rangle_{T}}\,2 [\arcsin(\sqrt{p_t})]|^{T}_0\nonumber\\
    &=-\frac{\hbar}{\langle\!\langle \Delta H^{\text{\rm nc}}_t \rangle\!\rangle_{T}}\left[ \arcsin(|\hspace{-1mm}\braket{\psi_0}{\psi_T}\hspace{-1mm}|)-\frac{\pi}{2}\right]\nonumber\\
    &=\frac{\hbar}{\langle\!\langle \Delta H^{\text{\rm nc}}_t \rangle\!\rangle_{T}}\,\arccos(|\hspace{-1mm}\braket{\psi_0}{\psi_T}\hspace{-1mm}|)\nonumber\\
    &=\frac{\hbar\, \Theta_{0T}}{\langle\!\langle \Delta H^{\text{\rm nc}}_t \rangle\!\rangle_{T}}\nonumber\tag{S7}
\end{align}

{\it \bf Example 1:} Consider a quantum system consisting of a qubit with an initial state of $|0\rangle$, whose evolution is governed by a non-optimal Hamiltonian $H= \hbar \, \hat{n}\cdot\Vec{\sigma}$, where $\hat{n}=(n_x,n_y,n_z)$ is a unit vector. To determine the quantum speed limit time $T_{\text{QSL}}^{\text{IMT}}$ required to reach the final state $|{\psi_{T}}\rangle=-in_z\ket{0}+(n_y-in_x)\ket{1}$ from the given initial state $|\psi_0\rangle$,
we require the following quantities: $\ket{\psi_t}=(\cos{t}-in_z\sin {t})\ket{0}+(n_y-in_x)\sin{t}\ket{1}$, $\Theta_{\rm 0T}=\arccos(\sqrt{n_z})$ and 
$\langle\!\langle \Delta H^{\text{\rm nc}}_t \rangle\!\rangle_{T}=\frac{2}{\pi}\arccos(\sqrt{n_z})$. After putting these quantities into Eq. (\ref{inexactQSL}), we find that $T_{\text{QSL}}^{\text{IMT}}=\frac{\pi}{2} $, which is equal to the actual evolution time. Hence, the bound~\eqref{inexactQSL} is tight and attainable for the given unitary evolution. As we have assumed a general time-independent Hamiltonian for the qubit system in this example, we can conclude that the bound given by Eq.~\eqref{inexactQSL} is tight for all unitary processes generated by time-independent Hamiltonians in such systems. This example can also be used to verify the claim of theorem~\ref{Th1}.\\

{\it \bf  Example 2:}
Consider a two-qubit system with an initial state $\ket{\psi_0}=\ket{00}$, whose time evolution is governed by a non-optimal Hamiltonian $H=\sigma_x\otimes I+I\otimes \sigma_x$ (we have taken $\hbar=1$). To determine the exact time required to reach the final state $\ket{11}$ from the given initial state $\ket{00}$, we need to calculate $\Delta H_t^{\rm nc}$ and $l(\phi_t)|_0^T$ (as defined in Theorem 2). The time evolved state of the quantum system at any time $t$ under the unitary $U_t=e^{-iHt}=(\cos{t} I-i\sigma_x \sin{t})\otimes(\cos{t} I-i\sigma_x \sin{t})$ is given as 
 \begin{align}
     \ket{\psi_t}=\cos^2{t}\ket{00}-i\sin{t}\cos{t}(\ket{01}+\ket{10})-\sin^2{t} \ket{11}.\tag{S8}
 \end{align}
The basis of the two-qubit Hilbert space is given by {$\ket{a_0}=\ket{00}$, $\ket{a_1}=\ket{01}$, $\ket{a_2}=\ket{10}$, $\ket{a_3}=\ket{11}$}. It is now straightforward to verify that the classical part of the driving Hamiltonian is zero, resulting in $\Delta H=\Delta H^{nc}=\sqrt{2}$. The length traced out by the vector $\ket{\phi_t}$ in the real Hilbert space, connecting the initial state $\ket{00}$ and final state $\ket{11}$, is given by
\begin{align}
l(\phi_t)|^{T}_0=\bigintss_{0}^{T}\!\!\!\sqrt{\sum_{i=0}^{3}\,\Bigl(\frac{d|c_i|}{dt}\Bigr)^2}\;dt
=\frac{\pi}{\sqrt{2}}\,,\tag{S9}
\end{align}
We have used $c_1=\cos^2{t}$, $c_2=c_3=i\sin{t}\cos{t}$, and $c_4=\sin^2{t}$ to determine $l(\phi_t)|_0^T$. The exact time required to reach the target state $\ket{11}$ from the initial state $\ket{00}$ via the given driving Hamiltonian is
\begin{align*}
        T=\frac{ l(\phi_t)\vert^T_0}{\Delta H^{\rm nc}} = \frac{\pi}{2}\tag{S10}.
\end{align*}

{\it \bf Detailed Proof of Corollary 1:}
To show that the bound in Eq.~\eqref{exactQSLdd} is saturated for self-inverse driving Hamiltonians, we first need to calculate $\langle\!\langle \Delta H^{\text{\rm nc}}_t \rangle\!\rangle_{T} = \frac{1}{T}\int_{0}^{T} \!\!\Delta H^{\text{nc}}\,dt$.
Recall
\begin{align}\label{cor1}
     \left(\Delta H^{\text{nc}}\right)^2&=(\Delta H)^2-(\Delta H^{\text{cl}})^2\nonumber\\
     &=\braket{\psi_t}{H^2|\psi_t}-\braket{\psi_t}{(H^{\text{cl}})^2|\psi_t}\tag{S11},
\end{align}
where to arrive at the second expression, we use the fact that $\braket{\psi_t}{H^{\text{cl}}|\psi_t}=\braket{\psi_t}{H|\psi_t}$~\cite{Hall2001}. From the definition of $H^{\text{cl}}$ in Eq.~\eqref{Sam2} where the eigenbasis is $\{\ket{\psi_k}\}_{k=0}^{d-1}$ 
we may write 
\begin{align}\label{cor2}
\braket{\psi_t}{(H^{\text{cl}})^2|\psi_t} = \frac{1}{4}\sum_{k=0}^{d-1}\frac{(\braket{\psi_k}{H|\psi_t}\braket{\psi_t}{\psi_k}+ \text{c.c.})^2}{|\braket{\psi_k}{\psi_t}|^2}\tag{S12},
\end{align}
where c.c.\ denotes the complex conjugate. Therefore, Eq.~(\ref{cor1}) reduces to
\begin{align}\label{cor3}
\left(\Delta H^{\text{nc}}\right)^2 = \braket{\psi_t}{H^2|\psi_t} - \sum_{k=0}^{d-1}\frac{(\braket{\psi_k}{H|\psi_t}\braket{\psi_t}{\psi_k}+ \text{c.c.})^2}{4|\braket{\psi_k}{\psi_t}|^2}.\tag{S13}
\end{align}

Now, for a self-inverse Hamiltonian, i.e., one where $H^2=I$, the unitary evolution operator may be expressed as $U_t=e^{-i{t}H}=\cos({t})I-i\sin({t})H$. Here we have taken $\hbar=1$ and $H$ is dimensionless.

\def\leaveout2{
Since
\begin{align}
\braket{\psi_k}{H|\psi_t}\braket{\psi_t}{\psi_k}+ \text{c.c.}

=2 \, \delta_{k0}\, \langle \psi_0\vert H \vert \psi_0\rangle ,
\end{align}
we find
\begin{align}
\left(\Delta H^{\text{nc}}\right)^2
=\frac{\cos^2 t \,\bigl(1-\langle \psi_0\vert H \vert \psi_0\rangle^2\bigr)}
{1-\sin^2 t\,\bigl(1-\langle \psi_0\vert H \vert \psi_0\rangle^2\bigr)} .\nonumber
\end{align}
}

Substituting $\ket{\psi_t}=U_t\ket{\psi_0}$ into Eq.~(\ref{cor3}) and after simplification we obtain
 \begin{align}\label{cor4}
     \Delta H^{\text{nc}}=\frac{\sqrt{1-\braket{\psi_0}{H|\psi_0}^2}\;\cos t}
     {\!\!\sqrt{1-\bigl({1-\langle\psi_0\vert H\vert\psi_0\rangle^2}\bigr)\sin^2 t}}.\tag{S14}
 \end{align}
Integrating this expression we find
\begin{align}\label{cor5}
     &\;\langle\!\langle \Delta H^{\text{nc}} \rangle\!\rangle_{T}
     = \frac{1}{T}\int_{0}^{T}\!\!\! \Delta H^{\text{nc}} \,dt\nonumber\\
     =&\;\frac{1}{T} \arcsin\left[\sqrt{1-\braket{\psi_0}{H|\psi_0}^2}\hspace{1mm}\sin{T}\right]\nonumber\\
     =&\;\frac{1}{T} \arccos\left[\sqrt{\cos^2 T+\langle\psi_0\vert H\vert\psi_0\rangle^2\sin^2 T}\right]\nonumber\\
    =&\;\frac{1}{T}\,\arccos\bigl(\vert \langle \psi_0\vert\psi_T\rangle\vert\bigr)
    \equiv\frac{1}{T}\,\Theta_{0T},\tag{S15}
 \end{align}
where the identity $\arcsin(x)=\arccos(\sqrt{1-x^2})$ has been used. Substituting Eq.~(\ref{cor5}) into the definition of $T_{\text{QSL}}^{\text{IMT}}$ into Eq. (\ref{inexactQSL}) of the main manuscript, we obtain $T_{\text{QSL}}^{\text{IMT}}=T$ and we see that the bound is saturated.

\end{document}